\documentclass[onecolumn, journal]{IEEEtran}  

\IEEEoverridecommandlockouts                              
\usepackage[utf8]{inputenc}
\usepackage{amsmath,bm}

\usepackage{amsfonts}
\usepackage{amssymb}
\usepackage{graphicx}
\usepackage{amsfonts}
\usepackage{hyperref}
\usepackage{caption}
\usepackage{subcaption}
\usepackage{balance}
\usepackage{flushend}

\usepackage{algorithm2e}
\usepackage{mathtools}
\usepackage{xcolor}
\usepackage{cite}
\usepackage[normalem]{ulem}
\usepackage{textcomp}

\newtheorem{Theorem}{Theorem}

\newtheorem{Lemma}{Lemma}
\newtheorem{Problem}{Problem}

\newtheorem{Remark}{Remark}
\newtheorem{Assumption}{Assumption}

\usepackage{tikz}
	\usetikzlibrary{shadows,shapes,arrows}
	\tikzstyle{frame} = [draw, -latex]
	\tikzstyle{line} = [draw]
	\tikzstyle{line2} = [draw, dashdotted]
	\tikzstyle{line3} = [draw, dashed]
	\tikzstyle{line3UD} = [draw, dashed]
	\tikzstyle{place} = [circle, draw=black, fill=white, thick, inner sep=2pt, minimum size=1mm]
	\tikzstyle{place2} = [circle, draw=black, fill=black, thick, inner sep=2pt, minimum size=1mm]
	\tikzstyle{placeRed} = [circle, draw=red, fill=red, thick, inner sep=2pt, minimum size=1mm]
	\tikzstyle{vertex} = [circle, draw=black, fill=black, thick, inner sep=2pt, minimum size=1mm]
\usepackage{tkz-euclide}
\usepackage{multirow}

\makeatletter
\def\algbackskip{\hskip-\ALG@thistlm}
\makeatother

\usepackage[switch,columnwise]{lineno}
\usepackage{setspace}

\title{\LARGE \bf Distributed MPC-based Coordination of Traffic Perimeter and Signal Control: A Lexicographic Optimization Approach}
\author{Viet Hoang Pham$^{1}$ and Hyo-Sung Ahn$^{2}$, {\it Senior Member,~IEEE} 
\thanks{This work was supported by the National Research Foundation of Korea (NRF) under the grant NRF- 2022R1A2B5B03001459.}
\thanks{\small $^{1}$Computer Science Department, Faculty of Information Technology, Posts and Telecommunications Institute of Technology, Hanoi, Vietnam. E-mail:  {vietpx@ptit.edu.vn}.}
\thanks{\small $^{2}$Department of Mechanical and Robotics Engineering, Gwangju Institute of Science and Technology, Gwangju, Korea. E-mails: {hyosung@gist.ac.kr}.}
}

\allowdisplaybreaks

\begin{document}
\maketitle 
\thispagestyle{empty}
\pagestyle{empty}

\begin{abstract}
This paper introduces a comprehensive strategy that integrates traffic perimeter control with traffic signal control to alleviate congestion in an urban traffic network (UTN). The strategy is formulated as a lexicographic multi-objective optimization problem, starting with the regulation of traffic inflows at boundary junctions to maximize the capacity while ensuring a smooth operation of the UTN. Following this, the signal timings at internal junctions are collaboratively optimized to enhance overall traffic conditions under the regulated inflows. The use of a model predictive control (MPC) approach ensures that the control solution adheres to safety and capacity constraints within the network.
To address the computational complexity of the problem, the UTN is divided into subnetworks, each managed by a local agent. A distributed solution method based on the alternating direction method of multipliers (ADMM) algorithm is employed, allowing each agent to determine its optimal control decisions using local information from its subnetwork and neighboring agents. Numerical simulations using VISSIM and MATLAB demonstrate the effectiveness of the proposed traffic control strategy.
\end{abstract}

\section{Introduction}
In recent years, urban traffic demand has escalated substantially due to rapid population growth and increased urbanization. However, the expansion of road infrastructure in urban areas has remained minimal or is infeasible due to limited space. As a result, the development of advanced traffic control strategies for urban traffic networks (UTNs) has become increasingly critical. Numerous strategies have been proposed to optimize the utilization of existing traffic infrastructure, with the aim of mitigating congestion and improving traffic flow efficiency \cite{MarkosPapageorgiou2003, LeiChen2016}.
For example, TRANSYT-7F model \cite{KCourage1991} utilizes macroscopic modeling and historical traffic data to generate coordinated, offline control strategies for intersections, optimizing overall network throughput and reducing delays. In an effort to enhance the performance of traffic control systems, several coordinated, traffic-responsive approaches have been introduced, including SCOOT \cite{PBHunt1982}, SCATS \cite{PRLowrie1982}, PRODYN \cite{JLFarges1983}, RHODES \cite{PMirchandani1998}, and OPAC \cite{NHGartner1983}. These systems leverage real-time traffic data obtained from road sensors to optimize traffic signal timing based on short-term traffic forecasts. Although these methods have shown effectiveness under certain conditions, they often suffer from exponential computational complexity. Furthermore, they fail to incorporate critical safety and physical constraints related to roads and junctions (e.g., they cannot guarantee that the number of vehicles on a road will always remain below its capacity). As a result, these methods are primarily effective in low to moderate traffic demand scenarios, with significant performance degradation observed during periods of high traffic demand.

Unlike the above methods \cite{KCourage1991, PBHunt1982, PRLowrie1982, JLFarges1983, PMirchandani1998, NHGartner1983}, Model Predictive Control (MPC)-based traffic control approach is uniquely capable of guaranteeing the operational constraints essential for urban traffic network management. Recent advancements in sensing and communication technologies have also spurred the development of MPC-based traffic control systems, which have attracted significant attention in the field \cite{BaoLinYe2019, RRNegenborn2008, KonstantinosAmpountolas2009, KonstantinosAmpountolas2010, TamasTettamanti2014, SteliosTimotheou2015, BaoLinYe2016, ZhaoXhou2017, PietroGrandinetti2018, ShuLin2011, NaWu2019, NaWu2020}. These MPC-based methods utilize state-space models to predict future traffic behavior and formulate optimal control strategies within a rolling horizon framework. By incorporating real-time traffic conditions and continuously updating model parameters at each control interval, MPC-based approach offers improved accuracy and reliability in traffic signal control, particularly under dynamic and high-demand traffic conditions.
In MPC traffic control approach, the store-and-forward model \cite{KonstantinosAmpountolas2009, KonstantinosAmpountolas2010, BaoLinYe2016}, the cell transmission model \cite{SteliosTimotheou2015, PietroGrandinetti2018}, and the S model \cite{ShuLin2011, NaWu2019, NaWu2020} are widely employed due to their simplicity and sufficient accuracy in describing traffic evolution over time and space. Depending on the employed traffic model, the MPC traffic signal control problems can be formulated as either convex optimization problems \cite{RRNegenborn2008, KonstantinosAmpountolas2009, KonstantinosAmpountolas2010, TamasTettamanti2014, SteliosTimotheou2015, BaoLinYe2016, ZhaoXhou2017, PietroGrandinetti2018} or mixed-integer linear problems \cite{ShuLin2011, NaWu2019, NaWu2020}. Convex optimization problems require less computational load to solve and are shown to be more feasible for application in large UTNs.

Depending on the underlying methods used to solve optimal control problems, MPC-based traffic signal control systems can be categorized as either centralized or distributed systems. In centralized systems, a single powerful computational unit is required to determine the optimal traffic signal timing plans for all signalized junctions in the UTN. However, the computational burden may become overwhelming for a single controller when the network consists of many roads and junctions. The MPC-based traffic signal control problem becomes large-scale, involving a substantial number of variables and constraints. Additionally, due to the spatial distribution of the UTN, collecting traffic data to formulate the MPC-based problem is challenging due to conflicts and unreliability in data exchange flows. Consequently, both distributed collection of traffic model parameters and distributed decision-making for traffic signal control are highly desirable.
In distributed systems, an UTN is considered as a union of multiple subnetworks, with each subnetwork managed by a local controller, also referred to as an agent in a multi-agent system. By decomposing a centralized MPC-based traffic signal control problem into a network of interconnected subproblems, various distributed optimization methods are implemented to reduce the total execution time through parallel computing. The dual-decomposition-based method \cite{StephenBoyd2004, RuggeroCarli2013} is the most widely applied approach for solving distributed MPC-based traffic signal control problems \cite{RRNegenborn2008, BaoLinYe2016,ZhaoXhou2017, PietroGrandinetti2018, VietHoangPham2022}. However, it is required that the cost function to be minimized must be strictly quadratic. In contrast, the alternating direction method of multipliers (ADMM) algorithm \cite{StephenBoyd2011, BrendanODonoghue2016, VietHoangPham2023} exhibits more favorable convergence properties for general convex optimization problems.

Even though MPC-based traffic signal control methods can provide optimal traffic signal timing plans, traffic congestion may become inevitable when the traffic inflows entering the UTN from boundary junctions exceed the capacity of the existing road infrastructure. In such cases, the MPC-based traffic signal control problems become unfeasible, indicating that no solution satisfies all the physical and safety constraints. To protect UTNs from becoming oversaturated, several gating strategies based on the macroscopic fundamental diagram (MFD) \cite{KonstantinosAboudolas2013, NikolasGeroliminis2013, JackHaddad2014, MohammadHajiahmadi2014, MehdiKeyvanEkbatani2015, HengDing2017, QianChen2020} have been proposed. These works focus on keeping network accumulation (i.e., the total number of vehicles in the UTN) around a predetermined critical value. However, the MFD is not applicable under heterogeneous traffic conditions and is highly sensitive to changes in traffic demands and the implemented traffic signal timing plan. As a result, these gating controllers can only prevent saturated conditions for predetermined traffic signal timing plans of internal junctions, but they do not maximize the capacity of the UTN.
red{In traffic control literature, several published studies have implemented perimeter control based on MFD combined with adaptive traffic signal control methods, such as Max Pressure (MP) \cite{DimitriosTsitsokas2023, HaoLiu2024}, to increase the capacity and throughput of traffic networks. In these studies, perimeter controller focuses on maintaining the accumulated number of vehicles in the network near critical points (determined through MFD analysis) to allow MP to operate under ideal conditions. However, this approach is only suitable for scenarios with low or moderate traffic demands. In cases of saturated conditions or even unsaturated conditions with high traffic demand, it results in long queues at boundary junctions. Here, the saturated state refers to the situation where the traffic network has reached its maximum capacity for handling traffic flow without causing prolonged congestion.}

In order to cope with various traffic conditions, this paper follows a MPC-based approach.
Unlike other MPC-based traffic signal control works that consider all exogenous traffic inflows as predetermined parameters, we assumes that the exogenous traffic flows entering the UTN from perimeter can be controlled, e.g., by adjusting the traffic signals at the boundary junctions. Following the key concept of traffic perimeter control, our approach also relies on holding vehicles at the perimeter to avoid internal congestion, but rather than imposing a predetermined bound, it seeks to minimize the number of vehicles retained at the perimeter while still safeguarding smooth network operation.
The proposed control strategy is designed to achieve two main objectives: 1) \textbf{Perimeter control:} maximize the total exogenous traffic inflows while ensuring a smooth operation of the UTN; 2) \textbf{Traffic signal control:} determine the optimal signal splits for all internal junctions to improve traffic conditions within the UTN while maintaining the maximum allowable exogenous traffic inflows. The resulting control problem is formulated using a lexicographic optimization approach. To achieve this, we modify the store-and-forward modeling method to formulate the traffic perimeter and signal control problems as convex optimization ones.
Lexicographic optimization, a subfield of multi-objective optimization, has garnered significant attention in the control system community \cite{DefengHe2021, AndreaPozzi2022, YaqingJv2023}. Numerous control problems, particularly in process systems, involve multiple, often conflicting objectives. Typically, these objectives have different levels of importance and are expressed through distinct cost functions. By incorporating the principles of lexicographic optimization into the MPC framework, a lexicographic MPC controller is developed by solving a hierarchy of single-objective MPC problems in a sequential manner. Compared to other multi-objective optimal methodologies, lexicographic MPC demonstrates superior handling of prioritized objectives.

Since the UTNs considered in traffic control usually include many roads and junctions, the sizes of their corresponding optimal control problems may be large.
To cope with this issue, we follow a distributed approach to solve both two MPC-based traffic control problems: perimeter control and signal control. Assuming that the UTN is controlled by multiple cooperative agents, we develop an iterative update scheme for each agent based on a proximal Alternating Direction Method of Multipliers (ADMM) algorithm. This scheme relies solely on local information and guarantees agents to cooperatively find the precise optimal solutions.

Our main contributions are summarized as follows:
\begin{itemize}
\item This study is the first to combine the concept of perimeter control and traffic signal control into a lexicographic MPC-based traffic control strategy to maximize the capacity of the UTN.
\item Presentation of the lexicographic MPC-based traffic control strategy is provided through employing a modified store-and-forward modeling method to formulate an MPC-based traffic perimeter control problem and an MPC-based traffic signal control problem.
\item A fully distributed solution method is developed for determining optimal control decisions of the lexicographic MPC-based traffic control strategy.
\item The effectiveness of the proposed control strategy is illustrated through numerical simulations in VISSIM and MATLAB.
\end{itemize}

The remainder of the paper is organized as follows.
Section II introduces a modified store-and-forward model approach to formulate MPC-based traffic control problems. In Section III, we describe the coordination of traffic perimeter and signal control as a lexicographic multi-objective optimization problem. This section also provides a network decomposition method and give the distributed control problem statement.
Then, we reformulate distributed versions of the MPC-based traffic control problems in Section IV. These problems are then solved by applying a distributed method based ADMM, which is presented in Appendix-A. Some simulation results are provided and analyzed in Section V to verify the effectiveness of our proposed traffic control strategy. Finally, Section VI concludes this paper.
\subsection*{Notations}
We use $\mathbb{R}$, $\mathbb{R}_{-}$, $\mathbb{R}^n$, and $\mathbb{R}^{m \times n}$ to denote the set of real numbers, the set of non-positive real numbers, the set of $n$-dimension real vectors, and the set of $m \times n$ matrices, respectively.
Denote by $\textbf{1}_n$ and $\textbf{0}_n$ the vectors in $\mathbb{R}^{n}$ whose all elements are $1$ and $0$, respectively.
Let $\textbf{I}_n$ be the identity matrix in $\mathbb{R}^{n \times n}$ and $\textbf{O}_{m \times n}$ be the zero matrix in $\mathbb{R}^{m \times n}$. When the dimensions are clear, the subscripts of the matrices $\textbf{I}_n, \textbf{O}_{m \times n}$, and the vector $\textbf{1}_n, \textbf{0}_n$ are ignored.
For a given set of $m$ vectors, $\mathcal{A} = \{\textbf{a}_1, \textbf{a}_2, \dots, \textbf{a}_m\}$, we use $\textrm{col }\mathcal{A}$ to denote the following column vector
\[\textrm{col }\mathcal{A} = \textrm{col}\{\textbf{a}_1, \textbf{a}_2, \dots, \textbf{a}_m\} = [\textbf{a}_1^T, \textbf{a}_2^T, \dots, \textbf{a}_m^T]^T.\]
Let $\mathcal{M} = \{\textbf{M}_1, \textbf{M}_2, \dots, \textbf{M}_m\}$ be the set of $m$ matrices, we use $\textrm{blkdiag }\textbf{M}$ to denote the block diagonal matrix created by aligning $\textbf{M}_1, \textbf{M}_2, \dots, \textbf{M}_m$ along its diagonal.
\section{Traffic Model for MPC approach}
\subsection{Graph Representation}
In examining the operation of an UTN, two primary entities are considered: junctions and roads. A junction is a location where traffic flows intersect from different directions, while a road connects two junctions and allows vehicles to move in a single, specified direction. Junctions can be further categorized based on their locations within the UTN as follows:
\begin{enumerate}
\item \textit{boundary} junctions are located at the periphery of the UTN, serving as entry and/or exit points for vehicles moving into or out of the network;
\item \textit{internal} junctions are situated within the interior of the UTN.
\end{enumerate}
Traffic flows at a junction are regulated by traffic signals to prevent conflicts between vehicles moving in different directions. These signals are organized into distinct phases, with each phase controlling a set of traffic flows that operate simultaneously. On a given road, different lanes may accommodate separated traffic flows, which can be assigned to different traffic signal phases. We define a \textit{road link} as a group of lanes on a road that are governed by the same traffic signal phase.
To have an illustrative representation of the UTN, we use a node to represent a junction and use an edge to represent a road link. Then the UTN is described by a directed graph $\mathcal{T} = (\mathcal{J}, \mathcal{R})$ with the nodes set $\mathcal{J}$ and the edge set $\mathcal{R}$.
Let $\mathcal{J}^I$ be the set of internal junctions, and $\mathcal{J}^B$ be the set of boundary junctions. So, we have $\mathcal{J} = \mathcal{J}^I \cup \mathcal{J}^B$ and $\mathcal{J}^I \cap \mathcal{J}^B = \emptyset$.

For a road link $z \in \mathcal{R}$, we use $\sigma(z)$ and $\tau(z)$ to denote its source junction and destination junction, respectively. Vehicles move from $\sigma(z)$ to $\tau(z)$ in the road link $z$. It is clear that $\sigma(z), \tau(z) \in \mathcal{J}$.
Let $\mathbb{L} \subset \mathcal{R}\times\mathcal{R}$ be the set of pairs of road links where $(z, w) \in \mathbb{L}$ if and only if $\tau(z) = \sigma(w)$ and vehicles are allowed to leave the road link $z$ to enter the road link $w$ via the junction $\tau(z)$. When $(z, w) \in \mathbb{L}$, the road link $w$ is a downstream neighbor of the road link $z$ and $z$ is an upstream neighbor of $w$. Denote by $\mathcal{N}_z^-$ and $\mathcal{N}_z^+$ the set of downstream neighbors and the set of upstream neighbors of the road link $z$, respectively. We have the following mathematical formulation:
\[\mathcal{N}_z^- = \{w: (z, w) \in \mathbb{L}\},\textrm{ and } \mathcal{N}_z^+ = \{w: (w, z) \in \mathbb{L}\}.\]
For a junction $J_v \in \mathcal{J}$, we denote by $\mathcal{R}_{J_v}^{in}$ and $\mathcal{R}_{J_v}^{out}$ the sets of its incoming road links and outgoing road links, respectively.
\[\mathcal{R}_{J_v}^{in} = \{z: \tau(z) = J_v\},\textrm{ and } \mathcal{R}_{J_v}^{out} = \{z: \sigma(z) = J_v\}.\]
It is easy to see that $\mathcal{N}_z^- \subset \mathcal{R}_{J_v}^{out}$ if $z \in \mathcal{R}_{J_v}^{in}$, $\mathcal{N}_z^+ \subset \mathcal{R}_{J_v}^{in}$ if $z \in \mathcal{R}_{J_v}^{out}$ and $\mathcal{N}_z^- = \emptyset$ if $\tau(z) \in \mathcal{J}^B$, $\mathcal{N}_z^+ = \emptyset$ if $\sigma(z) \in \mathcal{J}^B$.

For any two different road links $z, w \in \mathcal{R}$, $w$ is reachable from $z$ or $z$ can reach to $w$ if there exists at least one directed sequence of road links $\{u_1, u_2, \dots, u_n\}$ such that $u_1 = z, u_n = w$ and $\tau(u_k) = \sigma(u_{k+1}) \in \mathcal{J}, k = 1, \dots, n-1$.
A road link $z \in \mathcal{R}$ is reachable from and can reach to itself.
In this paper, we focus on an UTN whose graph representation satisfies the following assumption.
\begin{Assumption}\label{as_graph}
For every road link $z \in \mathcal{R}$, it is reachable from at least one source road link $w \in \mathcal{R}$, where $\sigma(w) \in \mathcal{J}^B$, and it can reach to at least one destination road link $u \in \mathcal{R}$, where $\tau(u) \in \mathcal{J}^B$.
\end{Assumption}
\subsection{Prediction model}
Similar as in \cite{VietHoangPham2023}, this paper follows the store-and-forward modeling approach to formulate MPC-based traffic control problems. This approach was initially proposed by Gazis et al. \cite{DenosCGazis1963} and widely used in urban traffic control research \cite{ChristinaDiakaki2002, KonstantinosAmpountolas2009, KonstantinosAmpountolas2010, BaoLinYe2016}. In this paper, we introduce some notations for perimeter control.
\subsubsection{Vehicle conservation equations}
Let $n_z(t)$ be the number of vehicles within the road link $z$ at time $tT$ where $T$ is the interval between two control time steps and $t = 0, 1, \dots$ is the time index.
Assume that the number $n_{z}(t)$ of vehicles contained in the road link $z$ at time $tT$ is known. Then, the traffic states $n_{z}(t + k + 1|t)$ at the time $(t + k + 1)T$, for $ k = 0, 1, \dots$, can be predicted by the following conservation law equation
\begin{align}
&n_{z}(t+k+1|t) = n_{z}(t) + \sum\limits_{l = 0}^{k}\Biggl\{e_z^{in}(t+l|t) - e_z^{out}(t+l|t)\nonumber\\
&+ \sum\limits_{w \in \mathcal{N}_z^+}r_{wz}(t+l|t)f_w^d(t+l|t)
- f_z^d(t+l|t) + f_z^{u}(t+l|t)\Biggr\}.\label{eq_traffic_state_n}
\end{align}
In \eqref{eq_traffic_state_n}, all the traffic flows (i.e.,$f_{z}^d(t+l|t)$, $e_{z}^{in}(t+l|t)$, $e_z^{out}(t+l|t)$ and $f_z^{u}(t+l|t)$) and turning rates ($r_{wz}, \forall w \in \mathcal{N}_z^+$) are defined in the time interval $[(t+l)T, (t+l+1)T]$.
For more specific, $f_{z}^d(t+l|t)$ is the planned downstream traffic flow of the road link $z$, which is defined as the number of vehicles entering to the destination junction $\tau(z)$; $e_{z}^{in}(t+l|t)$ and $e_z^{out}(t+l|t)$ are respectively the estimated numbers of vehicles starting and ending their trips at the road link $z$;  $f_z^{u}(t+l|t)$ is the number of vehicles entering into the road link $z$ from one boundary junction. Here, we consider conventionally that $f_z^{u}(t+l|t) > 0$ only if $\sigma(z) \in \mathcal{J}^B$ and $f_z^{u}(t+l|t) = 0, \forall l \ge 0,$ otherwise. Fig. \ref{fig_model} depicts traffic flows for one road link. The estimated turning rate $r_{wz}(t+l|t)$ corresponds to the percentage of the downstream traffic flow from the road link $w$ entering to its downstream neighbor $z \in \mathcal{N}_w^-$.
It is clear that $r_{zw}(t+k|t) \ge 0, \forall w \in \mathcal{N}_z^-$, and $\sum_{w \in \mathcal{N}_z^-}r_{zw}(t+k|t) = 1, \forall k \ge 0$.

\begin{figure}[htb]
\centering
\scalebox{0.68}{\begin{tikzpicture}
\draw[-] (2.75,1.25) -- (5.25,1.25); \draw[-] (5.25,1.25) -- (5.25,0.5);
\draw[-] (2.75,2.75) -- (5.25,2.75); \draw[-] (5.25,2.75) -- (5.25,3.5);
\draw[-] (6.75,1.25) -- (6.75,0.5); \draw[-] (6.75,1.25) -- (7.5,1.25);
\draw[-] (6.75,2.75) -- (6.75,3.5); \draw[-] (6.75,2.75) -- (7.5,2.75);
\draw[dashed] (0.5,2) -- (1.5,2); \draw[dashed] (2.5,2) -- (5.5,2); \draw[dashed] (6.5,2) -- (7.5,2);
\draw[dashed] (6,0.5) -- (6,1.5); \draw[dashed] (6,2.5) -- (6,3.5);
\node at (2,2.1) [] (node_i) {\large $B_v$};
\node at (4,1.75) [] (link_z) {\large link $z$};
\node at (5.75,1.95) [] (link_z) {$f_z^d$};
\node at (2.5,1.2) [] (link_w2) {\large $f_{z}^u$};
\draw[->,red] (1,1.6) -- (3,1.6);
\draw[->,blue] (4.65,1.6) -- (5.75,1.6);
\draw[->,green] (3.45,0.75) -- (3.45,1.45);
\draw[->,green] (4.65,1.45) -- (4.65,0.75);
\node at (3.2,0.5) [] (dz) {\large $e_z^{in}$};
\node at (4.85,0.5) [] (sz) {\large $e_z^{out}$};
\node at (4.5,-0.15) [] (title) {\large Traffic flows in a source road link.};
\end{tikzpicture}}
\vspace{10pt}

\scalebox{0.68}{\begin{tikzpicture}
\draw[-] (1.25,1.25) -- (0.5,1.25); \draw[-] (1.25,1.25) -- (1.25,0.5);
\draw[-] (1.25,2.75) -- (0.5,2.75); \draw[-] (1.25,2.75) -- (1.25,3.5);
\draw[-] (2.75,1.25) -- (2.75,0.5); \draw[-] (2.75,1.25) -- (5.25,1.25); \draw[-] (5.25,1.25) -- (5.25,0.5);
\draw[-] (2.75,2.75) -- (2.75,3.5); \draw[-] (2.75,2.75) -- (5.25,2.75); \draw[-] (5.25,2.75) -- (5.25,3.5);
\draw[-] (6.75,1.25) -- (6.75,0.5); \draw[-] (6.75,1.25) -- (7.5,1.25);
\draw[-] (6.75,2.75) -- (6.75,3.5); \draw[-] (6.75,2.75) -- (7.5,2.75);
\draw[dashed] (0.5,2) -- (1.5,2); \draw[dashed] (2.5,2) -- (5.5,2); \draw[dashed] (6.5,2) -- (7.5,2);
\draw[dashed] (2,0.5) -- (2,1.5); \draw[dashed] (2,2.5) -- (2,3.5);
\draw[dashed] (6,0.5) -- (6,1.5); \draw[dashed] (6,2.5) -- (6,3.5);
\node at (2,2.1) [] (node_i) {\large $J_v$};
\node at (4,1.75) [] (link_z) {\large link $z$};
\node at (5.75,1.95) [] (link_z) {$f_z^d$};
\node at (1.65,3.25) [] (link_w1) {\large $f_{w_1}^d$};
\node at (0.65,1.65) [] (link_w2) {\large $f_{w_2}^d$};
\node at (2.35,0.65) [] (link_w3) {\large $f_{w_3}^d$};
\draw[-,blue] (1.75,3) -- (1.75,1.8); \draw[->,blue] (1.75,1.8) -- (3,1.8);
\draw[->,blue] (1,1.6) -- (3,1.6);
\draw[-,blue] (2.25,0.95) -- (2.25,1.4); \draw[->,blue] (2.25,1.4) -- (3,1.4);
\draw[->,blue] (4.65,1.6) -- (5.75,1.6);
\draw[->,green] (3.45,0.75) -- (3.45,1.45);
\draw[->,green] (4.65,1.45) -- (4.65,0.75);
\node at (3.2,0.5) [] (dz) {\large $e_z^{in}$};
\node at (4.85,0.5) [] (sz) {\large $e_z^{out}$};
\node at (4.0,-0.15) [] (title) {\large Traffic flows in a road link connecting two internal junctions.};
\end{tikzpicture}}
\caption{Traffic flows in a road link $z \in \mathcal{L}$ for two cases: (top) $\sigma(z) = B_v \in \mathcal{J}^B$; (bottom) $\sigma(z) = J_v \in \mathcal{J}^I$ and $\mathcal{N}_z^+ = \{w_1, w_2, w_3\}$. The blue arrows correspond to downstream traffic flows, the green arrows are exogenous traffic flows of the road link $z$, and the red arrows are the exogenous traffic inflows from the boundary junction.}\label{fig_model}
\end{figure}

Consider the source road link $z$, where $\sigma(z) \in \mathcal{J}^B$, the number of vehicles that arrive to the boundary junction $\sigma(z)$ and want to enter into $z$ is called the traffic demand of this road link. We use $d_z(t+k|t)$ to denote the estimated traffic demand in the time interval $[(t+k)T, (t+k+1)T]$. The queue length at the time $(t+k+1)T$ for the source road link $z$ is predicted as
\begin{equation}\label{eq_queue_demand}
q_z(t+k+1|t) = q_z(t) + \sum_{l = 0}^{k}d_z(t+l|t) - \sum_{l = 0}^{k}f_{z}^{u}(t+l|t).
\end{equation}
where $q_z(t)$ is the number of vehicles waiting at the boundary junction $\sigma(z) \in \mathcal{J}^B$ to enter the road link $z$ at the time $tT$. For simplicity of notations, we consider that $d_z(k) = q_z(k) = f_z^u(k) = 0$ for all $k \ge 0$ if $\sigma(z) \notin \mathcal{J}^B$.
\subsubsection{Traffic signal in internal junctions}
The traffic signal timing plan of a signalized junction $J_v \in \mathcal{J}^I$ is determined by a sequence of traffic signal phases and their assigned splits.
We use $\mathcal{P}_{J_v}$ to denote the set of traffic signal phases of this junction.
Denote by $g_{p}(t+k|t)$ the traffic signal split, i.e., the green time, assigned to the phase $p \in \mathcal{P}_{J_v}$ in the time interval $[(t+k)T,(t+k+1)T]$.
To control the UTN as a whole system, the operations of all junctions need to be synchronized as they share a common cycle.
A unified cycle allows synchronization of green phases across intersections, reducing stop-and-go movements and preventing spillbacks. If each junction operated with its own cycle, phase misalignments would accumulate, leading to congestion and loss of effective capacity. The advantages of synchronous operation over
asynchronous operation have been shown in \cite{LucasBarcelosDeOliveira2010}.
For simplicity, let the cycle of all junctions be equal to the control time interval $T$.
The sequence of traffic signal phases in every junction is assumed to be predetermined suitably with its own structure.
The lost time $L_{J_v}$ is defined as the sum of time lengths for clearing vehicles between consecutive traffic signal phases in a junction $J_v \in \mathcal{J}^I$.
Then the following constraint needs to be satisfied
\begin{equation}\label{eq_signal_limit1}
\sum_{p \in \mathcal{P}_{J_v}} g_{p}(t+k|t) \le T - L_{J_v} := T_{J_v}.
\end{equation}
In addition, the traffic signal split of one phase is necessarily non-negative as
\begin{equation}\label{eq_signal_limit2}
g_p(t+k|t) \ge 0, \forall p \in \mathcal{P}_{J_v}.
\end{equation}
Let $\mathcal{P}_{z} \subsetneq \mathcal{P}_{J_v}$ be the set of traffic signal phases which give the right of way to the road link $z$ where $\tau(z) = J_v \in \mathcal{J}^I$.
Then the green time assigned to the road link $z$ in the time interval $[(t+k)T,(t+k+1)T]$ is $g_z(t+k|t) \le \sum_{p \in \mathcal{P}_{z}} g_{p}(t+k|t)$.
Since there may be more than one road link that can be activated in one traffic signal phase, we distinguish the green time $g_z(t+k|t)$ from the splits of phases in $\mathcal{P}_z$. Then even if two road links $w, z \in \mathcal{R}_{J_v}^{in}$ have the same set of the assigned traffic signal phases, i.e., $\mathcal{P}_{z} \equiv \mathcal{P}_{w}$, their assigned green time lengths can be different. This setup enhances the flexibility of the traffic model.
For example, two road links $z$ and $w$ are in the same traffic signal phase $p$ but one downstream neighbor of the road link $z$ is in saturated condition and should not receive more vehicles while all downstream neighbors of the road link $w$ have enough free space, it is allowed to set $g_z(t) = 0$ and $g_w(t) > 0$.
\begin{Remark}
In Appendix-C, we provide an illustrative example to make clear the notations for UTN graph representation, traffic signal phase, and the network decomposition (will be presented in Subsection III-C).
\end{Remark}
\subsection{Traffic flows}
In store-and-forward modeling approach \cite{DenosCGazis1963, ChristinaDiakaki2002}, the downstream traffic flow of the road link $z \in \mathcal{R}$ is determined as $f_{z}^d(t+k|t) = S_{z} g_z(t+k|t)$ where $S_{z}$ is the saturation flow of the road link $z$.
This implicitly assumes that the green time intervals are always fully utilized, i.e., there are enough vehicles awaiting in the road link $z \in \mathcal{R}$ in its green time duration.
To guarantee this requirement, in this paper, we assume that the time interval $T$ is chosen to be small enough and require the following constraint. 
\begin{equation}\label{eq_predictedTrafficFlow}
0 \le f_z^d(t+k|t) \le n_z(t+k|t) + f_z^u(t+k|t) + e_z(t+k|t),
\end{equation}
where $e_z(t+k|t) := e_z^{in}(t+k|t) - e_z^{out}(t+k|t)$.
In other words, the downstream traffic flow departing from the road link $z \in \mathcal{L}$ in the time interval $[(t+k)T, (t+k+1)T]$ is less than the summation of vehicles in this road link at the time $(t+k)T$ and the difference of its exogenous traffic flows in the time interval $[(t+k)T, (t+k+1)T]$.

Let $\overline{n_z}$ be the maximum number of vehicles the road link $z \in \mathcal{R}$ can contain without traffic jam. So, to avoid traffic jam, the upstream traffic flows entering into the road link $z$ are necessary to satisfy the constraint (\ref{eq_upstream_limit}a) if $\sigma(z) \in \mathcal{J}^I$ or the constraint (\ref{eq_upstream_limit}b) if $\sigma(z) \in \mathcal{J}^B$ as follows.
\begin{subequations}\label{eq_upstream_limit}
\begin{gather}
\sum\limits_{w \in \mathcal{N}_z^+}r_{wz}(t+k|t)f_{w}^d(t+k|t) \le \overline{n_z} - n_z(t+k|t) - e_z(t+k|t),\\
f_z^u(t+k|t) \le \overline{n_z} - n_z(t+k|t) - e_z(t+k|t).
\end{gather}
\end{subequations}
The constraints \eqref{eq_predictedTrafficFlow} and \eqref{eq_upstream_limit} are well-studied in the cell transmission model (CTM).
In addition, we define the following constraint
\begin{equation}\label{eq_signal_roadlink}
f_z^d(t+k|t) \le S_z\sum_{p \in \mathcal{P}_{z}} g_{p}(t+k|t), \textrm{ if } \tau(z) \in \mathcal{J}^I.
\end{equation}
for the relation between the downstream traffic flow of the road link $z$ and the total green time of its corresponding traffic signal phases.

Consider a source road link $z$ where $\sigma(z) \in \mathcal{J}^B$. Beside the constraint (\ref{eq_upstream_limit}b) for the exogenous traffic inflow, it is natural to add the following constraint for its queue length at the boundary junction.
\begin{equation}\label{eq_queue_limit}
q_z(t+k+1|t) \ge 0,
\end{equation}
for all $z$ where $\sigma(z) \in \mathcal{J}^B$ and $k \ge 0$.

For a destination road link $z$ where $\tau(z) \in \mathcal{J}^B$, as it can connect to another traffic network, we assume that there is an upper limitation for its downstream traffic flow in each time interval.
\begin{equation}\label{eq_outstream_roadlink}
f_z^d(t+k|t) \le \overline{f_{B_v}}(t+k|t), \textrm{ if } \tau(z) = B_v \in \mathcal{J}^B,
\end{equation}
where $\overline{f_{B_v}}(t+k|t)$ is given for all $k \ge 0$.

\section{MPC traffic control via perimeter and signal control coordination}
\subsection{Introduction to MPC-based traffic control}
In traffic control perspective, it is usually considered that: the numbers of vehicles contained in road links and the queue lengths at the boundary junctions, i.e, $n_z(\cdot) \forall z \in \mathcal{R}$ and $q_z(\cdot) \forall z: \sigma(z) \in \mathcal{J}^B$, are traffic states; the downstream traffic flows, the splits of traffic signal phases in internal junctions, and the incoming traffic flows of source road links, i.e., $f_{z}^d(\cdot) \forall z \in \mathcal{R}$, $g_p(\cdot) \forall p \in \mathcal{P}_{J_v}, J_v \in \mathcal{J}^I$, and $f_z^u(\cdot) \forall z: \sigma(z) \in \mathcal{J}^B$, are control decisions; the disturbance flows, the turning rates and the demands, i.e., $e_{z}^{in}(\cdot), e_{z}^{out}(\cdot), r_{wz}(\cdot), \forall w \in \mathcal{N}_z^+,  z \in \mathcal{R}$ and $d_z(\cdot) \forall z: \sigma(z) \in \mathcal{J}^B$, can be predicted or estimated and are referred as traffic model parameters.
MPC-based traffic control methods aim to find the optimal control plan for next $K$ control time steps by employing the current traffic states and predicted traffic model parameters in the control formulation. The objectives of MPC-based traffic control systems typically are: (i) maximizing the capacity of the existing transportation infrastructure; and (ii) guaranteeing the smooth operation of the UTN.

The effectiveness of control decisions is usually evaluated by the following performance indexes:
\begin{align}
\Phi^{(1)}(t) &= \sum_{k = 0}^{K - 1} \sum_{z: \sigma(z) \in \mathcal{J}^B} q_z(t+k+1|t), \label{eq_cost1}\\
\Phi^{(2)}(t) &= \sum_{k = 0}^{K - 1} \sum_{z \in \mathcal{R}} \left(n_z(t + k|t) - f_z^d(t + k|t)\right). \label{eq_cost2}\\
\Phi^{(3)}(t) &= \sum_{k = 0}^{K - 1} \sum_{z \in \mathcal{R}} \frac{(n_z(t + k + 1|t))^2}{\overline{n_z}}. \label{eq_cost3}
\end{align}
The first one measures the total queue lengths at boundary junctions. Minimizing $\Phi^{(1)}(t)$ coincides with the goal of reducing the number of vehicles waiting to enter into the UTN. This objective is motivated by policy and operational concerns. In practice, long queues at the perimeter are often unacceptable in practice due to fairness issues, environmental impacts, and local regulations. The quantity $n_z(t + k|t) - f_z^d(t + k|t)$ is the number of vehicles stuck in the road link $z$ in the time interval $[(t+k)T, (t+k+1)T]$; thus, $\Phi^{(2)}(t)$ is minimized to reduce the time delayed for vehicles in the UTN. The third performance index $\Phi^{(3)}(t)$ aims to balance the relative occupancy of road links, which plays an importance role in reducing the congestion risk for the UTN in saturated and oversaturated situations.
To enhance the reliability of the computed control decisions, the physical and safety constraints for all road links and junctions described in the previous section need to be guaranteed. 

Let $t$ be the current control time step.
The MPC-based traffic control problem at time $t$ is formulated as the following optimization problem.
\begin{subequations}\label{eq_problemTP}
\begin{align}
\min\textrm{ }& \Phi(t) = \sum\limits_{k=0}^{K-1} \hat{\Phi}(t+k|t)\\
\textrm{s.t. }& \eqref{eq_traffic_state_n}, \eqref{eq_predictedTrafficFlow}, \eqref{eq_upstream_limit}, \eqref{eq_signal_roadlink}, \eqref{eq_outstream_roadlink}, \forall z \in \mathcal{R}, k \in [0, K-1],\\
&\eqref{eq_queue_demand}, \eqref{eq_queue_limit}, \forall z: \sigma(z) \in \mathcal{J}^B, k \in [0, K-1],\\
&\eqref{eq_signal_limit1}, \eqref{eq_signal_limit2} \forall J_v \in \mathcal{J}^I, k \in [0, K-1].
\end{align}
\end{subequations}
where the cost function $\hat{\Phi}(t+k|t)$ is one or the combination of performance indexes for evaluation of traffic control decisions at the control time step $t+k$.
After solving the constrained optimization problem \eqref{eq_problemTP}, only the determined control decisions corresponding to the current control time step are implemented for controlling the UTN. They include traffic flows entering the UTN from boundary junctions and green time splits, i.e., $f_z^u(t|t) \forall z: \sigma(z) \in \mathcal{J}^B$, and $g_z(t|t) = \frac{1}{S_z}f_z^d(t|t) \forall z:\tau(z) \in \mathcal{J}^I$.
This procedure is repeated in every control time step, $t+1, t+2, \cdots$, not only to improve traffic conditions of the UTN but also to enhance the robustness of the control decisions with uncertainty in measurement and estimation of traffic model parameters. 

The number $K$ in the problem \eqref{eq_problemTP} is called the prediction horizon, which plays an important role in MPC formulation. A short horizontal length is more reliable than the long ones. It reduces computational complexity to make a real-time application feasible and provides a more robust performance to uncertainty, disturbances, and model mismatch. However, if $K$ is too small, it may cause negative aftereffects as a myopic control plan. So, the horizontal length needs to be selected suitably. In MPC-based traffic signal control \cite{BaoLinYe2019, RRNegenborn2008, KonstantinosAmpountolas2009, KonstantinosAmpountolas2010, TamasTettamanti2014, SteliosTimotheou2015, BaoLinYe2016, ZhaoXhou2017, PietroGrandinetti2018, ShuLin2011, NaWu2019, NaWu2020}, $K$ is usually chosen from $3$ to $20$ depending on the time interval $T$. In this paper, we consider $K = 3 \sim 6$ with $T = 60$ seconds.
\subsection{Lexicographic multi-objective MPC problem}
The target of the traffic control strategy discussed in this paper involves two main objectives: 1) \textit{Perimeter Control:} maximizing the capacity of the UTN while guaranteeing the smooth operation; and 2) \textit{Traffic Signal Control:} determining traffic signal timing plans for internal junctions of the UTN in order to reduce the congestion risk as well as the time delay of vehicles traversing it.
These two objectives can be represented by the function $\Phi^{(1)}(t)$ (given in \eqref{eq_cost1}) and the combination of $\Phi^{(2)}(t)$ and $\Phi^{(3)}(t)$ (given in \eqref{eq_cost2} and \eqref{eq_cost3}), respectively. Typically, these cost functions are somewhat conflicting. When the network receives too many vehicles, due to the lack of available space, movement will be hindered and the risk of congestion will increase. To reconcile these two objectives, a common approach is to combine them into a single objective function by assigning corresponding weights that reflect their relative importance. However, the process of selecting these weights is complex and leads to a trade-off.
An alternative approach employs the lexicographic ordering procedure, wherein control objectives are ranked to establish a hierarchy. The highest-priority objective is placed at the top of this hierarchy, while the lowest-priority objective is placed at the bottom. The lexicographic optimization approach eliminates the need for tuning and the relative selection of weights, thereby simplifying the optimization process. 

In this study, we prioritize maximizing the capacity of the UTN as our primary objective. By allowing the largest possible inflows into the traffic network without triggering internal congestion—thereby ensuring its smooth operation—this objective effectively contributes to enhancing the overall throughput, i.e., the total number of vehicles served by the UTN.
Beside the physical and safety constraints (\ref{eq_problemTP}b-\ref{eq_problemTP}d), we further add the following constraint for guaranteeing the smooth operation of the UTN.
\begin{equation}\label{eq_smooth_UTN}
n_z(t+k|t) - f_z(t+k|t) \le \gamma_z\overline{n_z}, \forall z \in \mathcal{R}, k \ge 0,
\end{equation}
where $\gamma_z$ is a given parameter for each road link $z \in \mathcal{R}$. The constraint \eqref{eq_smooth_UTN} is designed to ensure that the number of vehicles trapped in each road link during one control time interval remains below a specified quantity. This not only facilitates traffic mobility within the UTN but also indirectly ensures that road links have sufficient space to accommodate vehicles from their upstream neighbors. The rule for choosing the parameter $\gamma_z$ is not considered in this paper and is left as a future work.
Let $\Phi^{(PC)}$ and $\Phi^{(TSC)}$ be the cost functions to be minimized for two considered objectives, i.e., perimeter control and traffic signal control, respectively.
Our lexicographic MPC-based traffic control strategy begins by minimizing the cost function $\Phi^{(PC)}(t) = \Phi^{(1)}(t)$, subject to the constraints in (\ref{eq_problemTP}b-\ref{eq_problemTP}d) and \eqref{eq_smooth_UTN} first. Given that $\Phi^{(1)}(t)$ is a linear cost function, multiple optimal solutions may exist. Subsequently, a second-priority cost function $\Phi^{(TSC)}(t)$ is minimized, with an additional lexicographic constraint introduced. This constraint ensures that the primary objective function remains at its optimal value.

\begin{figure}[htb]
\centering
\scalebox{0.45}{\begin{tikzpicture}[
roundnode0/.style={circle, draw=white, fill=white,font=\fontsize{20}{20}\selectfont},
roundnode1/.style={circle, draw=green!60, fill=green!60, very thick, minimum size=1mm},
roundnode2/.style={circle, draw=red!60, fill=red!60, very thick, minimum size=1mm},
]
\node at (0,-0.5) [roundnode0] (ax_y1) {};
\node at (0,7) [roundnode0] (ax_y2) {};
\node at (0.35,7.3) [roundnode0] (ax_ylabel) {$\Phi^{(TSC)}$};
\node at (-0.5,0) [roundnode0] (ax_x1) {};
\node at (6.5,0) [roundnode0] (ax_x2) {};
\node at (7.25,0.15) [roundnode0] (ax_ylabel) {$\Phi^{(PC)}$};
\node at (0.85,5.5) [roundnode0] (objt1label_other) {$\mathcal{B}$};
\node at (1.5,5.5) [roundnode1] (objt1_other) {};
\node at (0.85,3.5) [roundnode0] (objt1label) {$\mathcal{A}$};
\node at (1.5,3.5) [roundnode1] (objt1) {};
\node at (1.5,0) [roundnode0] (objt1_ax) {};
\node at (1.5,-0.75) [roundnode0] (objt1_label) {$\Phi_{min}^{(PC)}$};
\node at (5.55,1.5) [roundnode0] (objt2label) {$\mathcal{C}$};
\node at (5,1.5) [roundnode1] (objt2) {};
\node at (0,1.5) [roundnode0] (objt2_ax) {};
\node at (-1,1.75) [roundnode0] (objt2_label) {$\Phi_{min}^{(TSC)}$};
\draw[->] (ax_y1.north) -- (ax_y2.south);
\draw[->] (ax_x1.east) -- (ax_x2.west);
\draw[-,dotted] (objt1_ax.north) -- (objt1.south);
\draw[-,dotted] (objt2_ax.east) -- (objt2.west);
\draw[-,red] (objt1_other.south) -- (objt1.north);
\draw[-,red] (objt1.south) .. controls +(down:8mm) and +(left:18mm) .. (objt2.west);
\end{tikzpicture}}
\caption{Illustration of multiobjective optimization.} \label{fig_lexicographic}
\end{figure}
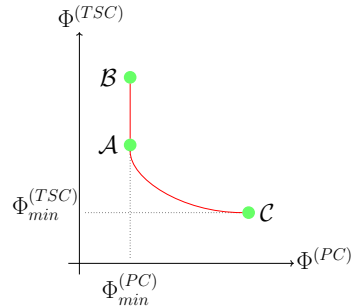

We use Fig. \ref{fig_lexicographic} to illustrate the difference between the lexicographic optimization approach and the weighted combination methods in multiobjective optimization. Denote by $\Phi_{min}^{(PC)}$ and $\Phi^{(TSC)}$ the minimum values of these two cost functions. Fig. \ref{fig_lexicographic} depicts the plane constructed by two cost functions $\Phi^{(PC)}$ and $\Phi^{(TSC)}$. For the lexicographic ordering implementation, the cost function $\Phi^{(PC)}$ is minimized first to determine the optimal cost value $\Phi_{min}^{(PC)}$. The solution found by this process is corresponding to one point in the line $\mathcal{A}\mathcal{B}$. The optimal solution of the lexicographic ordering approach is determined by minimizing the cost function $\Phi^{(TSC)}$ while guaranteeing $\Phi^{(PC)} = \Phi_{min}^{(PC)}$. Then it will find the point $\mathcal{A}$.
In the weighted combination approach, the optimal solution is to minimize a cost function in the form of $\Phi^{combined} = \vartheta\Phi^{(PC)} + \xi\Phi^{(TSC)}$ where $\vartheta$ and $\xi$ are positive weights. Depending on the chosen weights, the minimum value of the cost $\Phi^{combined}$, denoted by $\Phi_{min}^{combined}$, will be corresponding to one point in the curve $\mathcal{A}\mathcal{C}$. Though $\Phi_{min}^{combined}$ approaches to $\mathcal{A}$ as the value of $\frac{\vartheta}{\xi}$ goes to infinite, the computation load for the weighted combination approach with large value $\frac{\vartheta}{\xi}$ is significantly big compared to the lexicographic ordering approach (as shown in Subsection V-C).

The cost function $\Phi^{(TSC)}(t)$ reflects the improvement in traffic conditions of control decisions for the given total exogenous inflows. Usually, $\Phi^{(TSC)}(t)$ is chosen in the form of $\alpha\Phi^{(2)}(t) + \Phi^{(3)}(t)$ where $\alpha$ is a given weight. According to our experiment, the weight $\alpha$ should be selected in the range $[0.15, 0.3]$ to accelerate the traveling of vehicles in the unsaturated situations while motivating the balance of the road links occupancy in the saturated situations \cite{VietHoangPham2022, VietHoangPham2023}. Moreover, we add the term $\sum_{z: \sigma(z) \in \mathcal{J}^B}\beta(q_z(t+k+1|t))^2$, where $\beta$ is a small weight, for the equal distribution of queue lengths at boundary junctions. So, we have
\begin{align}
\Phi^{(TSC)}(t) = \sum_{k = 0}^{K - 1}\Biggl\{ \sum_{z \in \mathcal{R}} \Biggl(\frac{(n_z(t + k + 1|t))^2}{\overline{n_z}} + \alpha\biggl(n_z(t + k|t) - f_z^d(t + k|t)\biggr)\Biggr) + \sum_{z: \sigma(z) \in \mathcal{J}^B}\beta(q_z(t+k+1|t))^2\Biggr\}\label{eq_costTSC}
\end{align}

For detailed formulation, we define the optimization problem $\textbf{(PC)}$ corresponding to the highest-priority control objective as follows.
\begin{subequations}
\begin{align*}
\textbf{(PC)}: \textrm{ }\min &\sum\limits_{k=0}^{K-1} \sum_{z: \sigma(z) \in \mathcal{J}^B} q_z(t+k+1|t)\\
\textrm{ s.t. }& (\ref{eq_problemTP}b), (\ref{eq_problemTP}c), (\ref{eq_problemTP}d), \eqref{eq_smooth_UTN}.
\end{align*}
\end{subequations}
The problem $\textbf{(PC)}$ is referred to the MPC-based traffic perimeter control. 
Let $\Phi_{opt}^{(PC)}$ be the optimal cost of the problem $\textbf{(PC)}$. The optimal decisions corresponding to the lexicographic optimization approach are obtained by solving the following optimization problem $\textbf{(TSC)}$, which is referred to the MPC-based trafffic signal control problem.
\begin{subequations}
\begin{align*}
\textbf{(TSC)}:\textrm{ }\min &\Phi^{(TSC)}(t)\\
\textrm{ s.t. }& (\ref{eq_problemTP}b), (\ref{eq_problemTP}c), (\ref{eq_problemTP}d), \eqref{eq_smooth_UTN},\\
& \sum\limits_{k=0}^{K-1} \sum\limits_{z: \sigma(z) \in \mathcal{J}^B} q_z(t+k+1|t) = \Phi_{opt}^{(PC)}.
\end{align*}
\end{subequations}

Since an UTN typically consists of numerous roads and junction, the MPC-based traffic control problems, i.e., $\textbf{(PC)}$ and $\textbf{(TSC)}$, can become large and time-consuming to solve. This poses significant challenges in achieving a real-time application. To address this issue, a distributed solution method is necessary. By employing multiple computational units that collaborate in solving the MPC problems, the required computation time can be significantly reduced compared to a centralized control system. In the next subsection, we first propose a method to decompose the UTN and formulate the distributed traffic control problem.
\subsection{Distributed control problem}
Assume that the considered UTN is decomposed into $N$ interconnected subnetworks and each of them is controlled by one local controller. From a multiagent system perspective, these local controllers are called agents. They can work in parallel to share the total computational load and reduce the execution time. We use $\mathcal{S}_i$, $i = 1,\dots, N$, to denote one subnetwork and also its corresponding local controller.

Denote by $\mathcal{J}_i^{B}$ and $\mathcal{J}_i^{I}$ the sets of boundary and internal junctions of the subnetwork $\mathcal{S}_i$, respectively. The set $\mathcal{J}_i^{B}$ could be empty for some subnetworks. Then the set of all junctions in the subnetwork $\mathcal{S}_i$ is $\mathcal{J}_{i} = \mathcal{J}_i^{B} \cup \mathcal{J}_i^{I}$. The road links connecting junctions in the set $\mathcal{J}_{i}$ are called internal road links of the subnetwork $\mathcal{S}_i$. We define the set $\mathcal{R}_i$ as
\[\mathcal{R}_i = \{z \in \mathcal{R}: \sigma(z), \tau(z) \in \mathcal{J}_i\}.\]
Also, we define $\mathcal{R}_{ij}$ as the set of road links connecting from a junction in $\mathcal{S}_i$ to another in $\mathcal{S}_j$.
So, we have
\[\mathcal{R}_{ij} = \{z \in \mathcal{R}: \sigma(z) \in \mathcal{J}_i, \tau(z) \in \mathcal{J}_j\}.\]
If $\mathcal{R}_{ij} \neq \emptyset$, the agent $\mathcal{S}_i$ is called one neighbor of the agent $\mathcal{S}_j$ and vice versa. For each agent $\mathcal{S}_i$, we define the set of its neighbors as $\mathcal{N}_{\mathcal{S}_i} = \{\mathcal{S}_j: \mathcal{R}_{ij} \cup \mathcal{R}_{ji} \neq \emptyset\}$.
For better understanding of the aforementioned notations, we refer to an illustrative example in Appendix-C.

As each agent $\mathcal{S}_i$ is responsible in controlling the subnetwork $\mathcal{S}_i$, this agent needs to find the optimal control decisions corresponding the junction in the set $\mathcal{J}_i$.
They include: the splits of traffic phases in the internal junctions (i.e., $g_p(t+k|t), \forall p \in \mathcal{P}_{J_v}, J_v \in \mathcal{J}_i^I$ and $0 \le k \le K-1$), the downstream traffic flows of road links whose destination junctions are in $\mathcal{J}_i$ (i.e., $f_z^d(t+k|t), \forall z: \tau(z) \in \mathcal{J}_i$ and $0 \le k \le K-1$), and the traffic inflows from the boundary junctions (i.e., $f_z^u(t+k|t), \forall z: \sigma(z) \in \mathcal{J}_i^B$ and $0 \le k \le K-1$) if $\mathcal{J}_i^B \neq \emptyset$.
Consider the conservation law equation \eqref{eq_traffic_state_n} for the road link $z \in \mathcal{R}_{ij}$, we can observe that the traffic state $n_z(t+k+1|t)$ depends on the control decisions of both agents $\mathcal{S}_i$ and $\mathcal{S}_j$. To facilitate the cooperation of agents in solving MPC problems, it is necessary that the agent $\mathcal{S}_i$ can communicate to the agent $\mathcal{S}_j$ if $\mathcal{R}_{ij} \neq \emptyset$. So we make the following assumption
\begin{Assumption}\label{as_communication}
Each agent $\mathcal{S}_i$ can exchange information with all its neighbors in the set $\mathcal{N}_{\mathcal{S}_i}$. 
\end{Assumption}

Finally, we state the main control problem of this paper as follows.
\begin{Problem}\label{prob_main}
Design a distributed method for each agent $\mathcal{S}_i$ to find its corresponding control decisions in the optimal solution of the MPC problem $\textbf{(TSC)}$ while using only information belonging to itself or received from its neighbors in $\mathcal{N}_{\mathcal{S}_i}$.
\end{Problem}
\section{Proposed strategy}
In this section, we adopt a well-established approach to develop a distributed method for addressing Problem \ref{prob_main}. Initially, we formulate distributed optimization problems that are equivalent to $\textbf{(PC)}$ and $\textbf{(TSC)}$. Then, the optimal control decisions are obtained by using Algorithm \ref{alg_proposed_opt} (given in Appendix-A) as the distributed method for agents to cooperatively solve these distributed optimization problems.
\subsection{Distributed solution method for solving $\textbf{(PC)}$}
To formulate a distributed version of the problem $\textbf{(PC)}$, it is essential to define the local control variables for each agent. Since the agent $\mathcal{S}_i$ is responsible for controlling the internal junctions and boundary junctions within its subnetwork, the traffic states of the internal road links in the set $\mathcal{R}_i$ (i.e., $n_z(t+k+1|t) \forall z \in \mathcal{R}_i$ and $q_z(t+k+1|t) \forall z: \sigma(z) \in \mathcal{J}_i^B$), and the control decisions corresponding to the junctions in the set $\mathcal{J}_i$ (i.e., $g_p(t+k|t) \forall p \in \mathcal{P}_{J_v}, J_v \in \mathcal{J}_i^I$, $f_z^d(t+k|t) \forall z: \tau(z) \in \mathcal{J}_i^I$ and $f_z^u(t+k|t) \forall z: \sigma(z) \in \mathcal{J}_i^B$) are naturally considered as the local variables of the agent $\mathcal{S}_i$. For road links connecting subnetwork $\mathcal{S}_i$ to subnetwork $\mathcal{S}_j$, we assume that both agents $\mathcal{S}_i$ and $\mathcal{S}_j$ estimate their traffic states and downstream traffic flows. So, we have the stacked vector of all local variables of the agent $\mathcal{S}_i$ as follows.
\begin{equation}
\hat{\textbf{x}}_i = \textrm{col}\left\{\textrm{col}\left\{\hat{\textbf{n}}_{i,k}, \hat{\textbf{q}}_{i,k}, \hat{\textbf{f}}_{i,k}^d, \hat{\textbf{f}}_{i,k}^u, \hat{\textbf{g}}_{i,k}\right\}: k \in [0,K-1]\right\}.
\end{equation}
where $\hat{\textbf{n}}_{i,k} = \textrm{col}\Bigl\{n_z(t+k+1|t): \sigma(z) \in \mathcal{J}_i \textrm{ or } \tau(z) \in \mathcal{J}_i\Bigr\}$, $\hat{\textbf{q}}_{i,k} = \textrm{col}\Bigl\{q_z(t+k+1|t): \sigma(z) \in \mathcal{J}_i^B\Bigr\}$, $\hat{\textbf{f}}_{i,k}^d = \textrm{col}\Bigl\{f_z^d(t+k|t): \tau(z) \in \mathcal{J}_i \textrm{ or } \sigma(z) \in \mathcal{J}_i\Bigr\}$, $\hat{\textbf{f}}_{i,k}^u = \textrm{col}\Bigl\{f_z^u(t+k|t): \sigma(z) \in \mathcal{J}_i^B\Bigr\}$ and $\hat{\textbf{g}}_{i,k} = \textrm{col}\Bigl\{g_p(t+k|t): p \in \bigcup_{J_v \in \mathcal{J}_i^I} \mathcal{P}_{J_v}\Bigr\}$.

With the above definition of local variable vector, it can be found a matrix $\textbf{A}_i$ and a vector $\textbf{a}_i$ for each agent $\mathcal{S}_i$ such that the equality constraints \eqref{eq_traffic_state_n} and \eqref{eq_queue_demand} for all road links having source junctions in $\mathcal{J}_i$ over the time horizon $[t+1, t+K]$ can be presented in the following compact form
\begin{equation}\label{eq_local_eqconstr}
\textbf{A}_i\hat{\textbf{x}}_i = \textbf{a}_i, \forall i = 1, \dots, N.
\end{equation}
Similarly, the inequality constraints (\ref{eq_signal_limit1}-\ref{eq_signal_limit2}), $\forall J_v \in \mathcal{J}_i^I$, and (\ref{eq_predictedTrafficFlow}-\ref{eq_outstream_roadlink}, \ref{eq_smooth_UTN}), $\forall z: \tau(z) \in \mathcal{J}_i$, can be staked into the form of
\begin{equation}\label{eq_local_ineqconstr}
\textbf{B}_i\hat{\textbf{x}}_i \le \textbf{b}_i, \forall i = 1, \dots, N.
\end{equation}
with suitable matrix $\textbf{B}_i$ and vector $\textbf{b}_i$. So, the set of all constraints in the problem $\textbf{(PC)}$, i.e., (\ref{eq_problemTP}b-\ref{eq_problemTP}d) and \eqref{eq_smooth_UTN}, is equivalent to the set of constraints \eqref{eq_local_eqconstr}, \eqref{eq_local_ineqconstr}. In addition, we can find a vector $\textbf{c}_i$ for each agent $\mathcal{S}_i$ such that
\begin{equation}
\textbf{c}_i^T\hat{\textbf{x}}_i = \sum\limits_{k = 0}^{K-1}\sum\limits_{z: \sigma(z) \in \mathcal{J}_i^B} q_z(t+k+1).
\end{equation}
For the subnetwork $\mathcal{S}_i$ where $\mathcal{J}_i^B = \emptyset$, the vector $\textbf{c}_i$ is a zero vector.
Consider the road link $z \in \mathcal{R}_{ij}$, its traffic state $n_z(t+k+1|t)$ and its downstream traffic flow $f_z^d(t+k|t)$ are local variables of both agents $\mathcal{S}_i$ and $\mathcal{S}_j$ for all $k \ge 0$. It is natural to require that the elements corresponding to $n_z(t+k+1|t)$ and $f_z^d(t+k|t)$ in $\hat{\textbf{x}}_i$ are equal to the ones in $\hat{\textbf{x}}_j$. These relations are called coupling constraints between two agents $\mathcal{S}_i$ and $\mathcal{S}_j$. Let $\textbf{A}_{ij}$ be the matrix such that all coupling constraints between two agents $\mathcal{S}_i$ and $\mathcal{S}_j$ can be represented by the equation \eqref{eq_coupl_eqconstr}.
\begin{equation}\label{eq_coupl_eqconstr}
\textbf{A}_{ij}\hat{\textbf{x}}_i + \textbf{A}_{ji}\hat{\textbf{x}}_j = \textbf{0}, \forall \mathcal{S}_j \in \mathcal{N}_{\mathcal{S}_i}, i = 1, \dots, N.
\end{equation}

In summary, we have the following optimization problem which is a distributed version for the problem $\textbf{(PC)}$.
\begin{equation}\label{eq_PC_compacted}
\min\limits_{\hat{\textbf{x}}_i, \forall i = 1, \dots, N} \sum_{i = 1}^{N}\textbf{c}_i^T\hat{\textbf{x}}_i \textrm{ s.t. } \eqref{eq_local_eqconstr}, \eqref{eq_local_ineqconstr}, \eqref{eq_coupl_eqconstr}.
\end{equation}
In this optimization problem, the cost function is the summation of separated parts, each of them belongs to one agent. The equations \eqref{eq_local_eqconstr} and \eqref{eq_local_ineqconstr} are called local constraints which depend only on local information of individual agents. Meanwhile, the equation \eqref{eq_coupl_eqconstr} is called coupled constraint since it presents the relation between two neighboring agents.

As the optimization problem \eqref{eq_PC_compacted} has the similar form as in the equation \eqref{eq_distributedOpt}, its optimal solution can be found by applying Algorithm \ref{alg_proposed_opt}. Let $\hat{\textbf{x}}_i^{(PC)}$ be the output of the agent $\mathcal{S}_i$ when using Algorithm \ref{alg_proposed_opt} to solve the optimization problem \eqref{eq_PC_compacted}. We have
\begin{equation}\label{eq_solve_PC}
\hat{\textbf{x}}_i^{(PC)} = \textrm{DistSol}(\textbf{0}, \textbf{c}_i,\textbf{A}_i, \textbf{a}_i, \textbf{B}_i, \textbf{b}_i, \{\textbf{A}_{ij}: \mathcal{S}_j \in \mathcal{N}_{\mathcal{S}_i}\}, \mathcal{N}_{\mathcal{S}_i})
\end{equation}
for all $i = 1, 2, \dots, N$. The optimal cost value for the problem $\textbf{(PC)}$ can be computed as $\Phi_{opt}^{(PC)} = \sum_{i = 1}^{N}\textbf{c}_i^T\hat{\textbf{x}}_i^{(PC)}$.
\subsection{Distributed solution method for solving $\textbf{(TSC)}$}
For each agent $\mathcal{S}_i$, it can be determined a matrix $\textbf{H}_i$ and a vector $\textbf{h}_i$ such that \[\frac{1}{2}\hat{\textbf{x}}_i^T\textbf{H}_i\hat{\textbf{x}}_i + \textbf{h}_i^T\hat{\textbf{x}}_i = \sum_{k = 0}^{K - 1}\Biggl\{ \sum_{z: \sigma(z) \in \mathcal{J}_i} \Biggl(\frac{(n_z(t + k + 1|t))^2}{\overline{n_z}} + \alpha\biggl(n_z(t + k|t) - f_z^d(t + k|t)\biggr)\Biggr) + \sum_{z: \sigma(z) \in \mathcal{J}_i^B}\beta(q_z(t+k+1|t))^2\Biggr\}.\]
Here the set $\{z: \sigma(z) \in \mathcal{J}_i\}$ consists of road links whose source junctions are in the set $\mathcal{J}_i$. So, the cost function $\Phi^{(TSC)}(t)$ has a separated form as $\Phi^{(TSC)}(t) = \sum_{i = 1}^{N}\left(\frac{1}{2}\hat{\textbf{x}}_i^T\textbf{H}_i\hat{\textbf{x}}_i + \textbf{h}_i^T\hat{\textbf{x}}_i\right)$.
Then the optimization problem $\textbf{(TSC)}$ can be rewritten in the compact form as follows.
\begin{subequations}\label{eq_TSC_compact}
\begin{align}
\min_{\hat{\textbf{x}}_i, \forall i = 1, \dots, N}& \sum_{i = 1}^{N}\left(\frac{1}{2}\hat{\textbf{x}}_i^T\textbf{H}_i\hat{\textbf{x}}_i + \textbf{h}_i^T\hat{\textbf{x}}_i\right)\\
\textrm{s.t. }& \eqref{eq_local_eqconstr}, \eqref{eq_local_ineqconstr}, \eqref{eq_coupl_eqconstr},\\
& \sum_{i = 1}^{N}\textbf{c}_i^T\hat{\textbf{x}}_i = \sum_{i = 1}^{N}\textbf{c}_i^T\textbf{x}_i^{(PC)}.
\end{align}
\end{subequations}
The equation (\ref{eq_TSC_compact}c) consists of information from many agents which may not communicate directly. In order to use Algorithm \ref{alg_proposed_opt} to find the optimal solution of the problem $\textbf{(TSC)}$, we need to transform the feasible set of \eqref{eq_TSC_compact} into the combination of local constraints and coupling constraints between neighboring agents.

Let each agent $\mathcal{S}_i$ estimate virtual variables $\tilde{v}_{ij}, \forall \mathcal{S}_j \in \mathcal{N}_{\mathcal{S}_i}$, such that
\begin{subequations}\label{eq_temp}
\begin{align}
\textbf{c}_i^T\hat{\textbf{x}}_i + \sum_{\mathcal{S}_j \in \mathcal{N}_{\mathcal{S}_i}}\tilde{v}_{ij} = \textbf{c}_i^T\textbf{x}_i^{(PC)}, \forall i = 1, 2, \dots, N,\\
\tilde{v}_{ij} + \tilde{v}_{ji} = 0, \forall \mathcal{S}_j \in \mathcal{N}_{\mathcal{S}_i}, \forall i = 1, 2, \dots, N.
\end{align}
\end{subequations}
Define the variable vector $\tilde{\textbf{x}}_i$ as
\begin{equation}
\tilde{\textbf{x}}_i = \left[\hat{\textbf{x}}, \textrm{col}\{\tilde{v}_{ij}: \mathcal{S}_j \in \mathcal{N}_{\mathcal{S}_i}\}\right], \forall i = 1, 2, \dots, N.
\end{equation}
Let the matrix $\tilde{\textbf{A}}_{ij}, \forall \mathcal{S}_i \in \mathcal{N}_{\mathcal{S}_i}, \forall i = 1, \dots, N$, be defined by
\[\tilde{\textbf{A}}_{ij}\tilde{\textbf{x}}_i = \textrm{col}\{\textbf{A}_{ij}\hat{\textbf{x}}_i, \tilde{v}_{ij}\}.\]
We also define the matrices $\tilde{\textbf{H}}_i = \textrm{blkdiag}\{\textbf{H}_i, \textbf{O}\}$, $\tilde{\textbf{A}}_i = \left[\begin{matrix} \textbf{A}_i & \textbf{O}\\ \textbf{c}_i^T & \textbf{1}^T\end{matrix}\right]$, $\tilde{\textbf{B}}_i = \left[\textbf{B}_i, \textbf{O}\right]$, and the vector $\tilde{\textbf{h}}_i = \textrm{col}\{\textbf{h}_i, \textbf{0}\}$ such that $\tilde{\textbf{x}}_i^T\tilde{\textbf{H}}_i\tilde{\textbf{x}}_i = \hat{\textbf{x}}_i^T\textbf{H}_i\hat{\textbf{x}}_i$, $\tilde{\textbf{A}}_i\tilde{\textbf{x}}_i = \textrm{col}\{\textbf{A}_i\hat{\textbf{x}}_i, \textbf{c}_i^T\hat{\textbf{x}}_i + \sum_{\mathcal{S}_j \in \mathcal{N}_{\mathcal{S}_i}}\tilde{v}_{ij}\}$, $\tilde{\textbf{B}}_i\tilde{\textbf{x}}_i = \textbf{B}_i\hat{\textbf{x}}_i$, and $\tilde{\textbf{h}}_i^T\tilde{\textbf{x}}_i = \textbf{h}_i^T\hat{\textbf{x}}_i$.
Replacing the constraint (\ref{eq_TSC_compact}c) by the equations in \eqref{eq_temp}, we have the following optimization problem:
\begin{subequations}\label{eq_TSC_tosolve}
\begin{align}
\min_{\tilde{\textbf{x}}_i, \forall i}& \sum_{i = 1}^{N}\left(\frac{1}{2}\tilde{\textbf{x}}_i^T\tilde{\textbf{H}}_i\tilde{\textbf{x}}_i + \tilde{\textbf{h}}_i^T\tilde{\textbf{x}}_i\right)\\
\textrm{s.t. }& \tilde{\textbf{A}}_i\tilde{\textbf{x}}_i = \tilde{\textbf{a}}_i, \forall i = 1, \dots, N,\\
& \tilde{\textbf{B}}_i\tilde{\textbf{x}}_i \le \tilde{\textbf{b}}_i, \forall i = 1, \dots, N,\\
& \tilde{\textbf{A}}_{ij}\tilde{\textbf{x}}_i + \tilde{\textbf{A}}_{ji}\tilde{\textbf{x}}_j = \textbf{0}, \forall \mathcal{S}_j \in \mathcal{N}_{\mathcal{S}_i}, i = 1, \dots, N.
\end{align}
\end{subequations}
where $\tilde{\textbf{a}}_i = \textrm{col}\{\textbf{a}_i, \textbf{c}_i^T\textbf{x}_i^{(PC)}\}$ and $\tilde{\textbf{b}}_i = \textbf{b}_i, \forall i = 1, \dots, N$. Each equation in (\ref{eq_temp}a) is considered as a local constraint for an individual agent and embedded into (\ref{eq_TSC_tosolve}b). Meanwhile, each equation in (\ref{eq_temp}b) corresponding to agent $\mathcal{S}_i$ and agent $\mathcal{S}_j$ is a coupling constraint between these two agents and embedded into (\ref{eq_TSC_tosolve}d).
The following lemma verifies that the optimal solution of the problem $\textbf{(TSC)}$ can be found by solving the optimization problem \eqref{eq_TSC_tosolve}.
\begin{Lemma}\label{lm_TSC}
Let $\tilde{\textbf{x}}^{opt} = \textrm{col}\{\tilde{\textbf{x}}_i^{opt}: 1 \le i \le N\}$ be the optimal solution of the problem \eqref{eq_TSC_tosolve}. Then the vector $\hat{\textbf{x}}^{opt} = \textrm{col}\{\textbf{D}_i\tilde{\textbf{x}}_i^{opt}: 1 \le i \le N\}$, where $\textbf{D}_i = [\textbf{I}, \textbf{O}]$ such that $\textbf{D}_i\tilde{\textbf{x}}_i = \hat{\textbf{x}}_i, \forall i = 1, 2, \dots, N$, is the optimal solution of the problem \eqref{eq_TSC_compact}.
\end{Lemma}
\begin{proof}
See Appendix-B.
\end{proof}
As the optimization problem \eqref{eq_TSC_tosolve} has the form of \eqref{eq_distributedOpt}, its optimal solution can be found by applying Algorithm \ref{alg_proposed_opt}. Let $\tilde{\textbf{x}}_i^{(TSC)}$ be the output of this process as
\begin{equation}\label{eq_solve_TSC}
\tilde{\textbf{x}}_i^{(TSC)} = \textrm{DistSol}(\tilde{\textbf{H}}_i, \tilde{\textbf{h}}_i,\tilde{\textbf{A}}_i, \tilde{\textbf{a}}_i, \tilde{\textbf{B}}_i, \tilde{\textbf{b}}_i, \{\tilde{\textbf{A}}_{ij}: \mathcal{S}_j \in \mathcal{N}_{\mathcal{S}_i}\}, \mathcal{N}_{\mathcal{S}_i})
\end{equation}
for all $i = 1, 2, \dots, N$.
\subsection{Strategy for Optimal Traffic Control}
\begin{figure}[htb]
\centering
\scalebox{0.85}{\begin{tikzpicture}[
textnode/.style={rectangle, fill=white, opacity=0.9, minimum size=1mm, text=blue,text opacity=1,font=\fontsize{20}{20}\selectfont},
processnode/.style={rectangle, draw=black, fill=white,  very thick, minimum size=5mm, text=black,text opacity=1,font=\fontsize{10}{10}\selectfont},
initendnode/.style={ellipse, draw=black, fill=white,  very thick, minimum size=5mm, text=black,text opacity=1,font=\fontsize{10}{10}\selectfont},
]
\node at (4,6) [textnode] (nStart_temp) {};
\node at (0,6) [initendnode] (nStart) {Begin control step};
\node at (0,5) [processnode] (nDeterminePC) {Determine \textcolor{blue}{$\textbf{c}_i, \textbf{A}_i, \textbf{a}_i, \textbf{B}_i, \textbf{b}_i, \textbf{A}_{ij}'s$}};
\node at (0,4) [processnode] (nSolvePC) {$\hat{\textbf{x}}_i^{(PC)} \gets$ \textcolor{blue}{\eqref{eq_solve_PC}}};
\node at (0,3) [processnode] (nDetermineTSC) {Determine \textcolor{blue}{$\tilde{\textbf{H}}_i, \tilde{\textbf{H}}_i, \tilde{\textbf{A}}_i, \tilde{\textbf{a}}_i, \tilde{\textbf{B}}_i, \tilde{\textbf{b}}_i, \tilde{\textbf{A}}_{ij}'s$}};
\node at (0,2) [processnode] (nSolveTSC) {$\tilde{\textbf{x}}_i^{(TSC)} \gets$\textcolor{blue}{\eqref{eq_solve_TSC}}};
\node at (0,1) [processnode] (nDeploy) {Apply optimal control decisions};
\node at (0,0) [initendnode] (nEnd) {Wait to next step};
\node at (4,0) [textnode] (nEnd_temp) {};
\draw[->,{line width=1pt},black] (nStart.south)--(nStart.south|-nDeterminePC.north);
\draw[->,{line width=1pt},black] (nDeterminePC.south)--(nDeterminePC.south|-nSolvePC.north);
\draw[->,{line width=1pt},black] (nSolvePC.south)--(nSolvePC.south|-nDetermineTSC.north);
\draw[->,{line width=1pt},black] (nDetermineTSC.south)--(nDetermineTSC.south|-nSolveTSC.north);
\draw[->,{line width=1pt},black] (nSolveTSC.south)--(nSolveTSC.south|-nDeploy.north);
\draw[->,{line width=1pt},black] (nDeploy.south)--(nDeploy.south|-nEnd.north);
\draw[-,{line width=1pt},black] (nEnd.east)--(nEnd_temp.west);
\draw[-,{line width=1pt},black] (nEnd_temp.west)--(nStart_temp.west);
\draw[->,{line width=1pt},black] (nStart_temp.west)--(nStart.east);
\end{tikzpicture}}
\caption{Diagram for proposed control strategy.} \label{fig_diagram}
\end{figure}
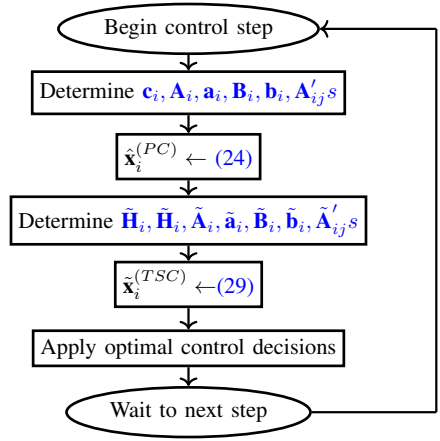
We summarize the proposed traffic control strategy for each agent $\mathcal{S}_i$ as the diagram in Fig. \ref{fig_diagram}. At the beginning of each control time step, the agent $\mathcal{S}_i$ uses sensors and appropriate estimation methods to measure and estimate the current traffic states, i.e., $n_z(t)$ and $q_z(t)$ if $\sigma(z) \in \mathcal{J}_i^B$ for all road links in $\{z: \sigma(z) \in \mathcal{J}_i \textrm{ or } \tau(z) \in \mathcal{J}_i\}$, as well as the parameters in traffic model. Then this agent can determine its local information in the formulation of the problem $\textbf{(PC)}$, stacked by the matrices $\textbf{A}_i, \textbf{B}_i, \textbf{A}_{ij} \forall \mathcal{S}_j \in \mathcal{N}_{\mathcal{S}_i}$, and the vectors $\textbf{a}_i, \textbf{b}_i$.
Following Algorithm \ref{alg_proposed_opt}, the agent $\mathcal{S}_i$ cooperates with its neighbors in $\mathcal{N}_{\mathcal{S}_i}$ to find the optimal solution of the problem \eqref{eq_PC_compacted} as $\hat{\textbf{x}}_i^{(PC)}$ in \eqref{eq_solve_PC}. After that, this agent formulates its corresponding parts in the problem \eqref{eq_TSC_tosolve} and employs Algorithm \ref{alg_proposed_opt} to cooperatively solve \eqref{eq_TSC_tosolve}. The solution $\tilde{\textbf{x}}_i^{(TSC)}$ in \eqref{eq_solve_TSC} corresponding to the local part of the agent $\mathcal{S}_i$ in the optimal solution of lexicographic MPC traffic control strategy is obtained. Only the control decisions for the current control time step, i.e., $g_p(t|t) \forall p: p \in \mathcal{P}_{J_v}, J_v \in \mathcal{J}_i^I$, $f_z(t|t) \forall z: \tau(z) \in \mathcal{J}_i^I$, and $f_z^u(t|t) \forall z: \sigma(z) \in \mathcal{J}_i^B$, are applied for controlling subnetwork $\mathcal{S}_i$. Then agents wait until the next control step.
\section{Simulations}
This section aims to show the effectiveness of the proposed traffic control strategy. Its performance in controlling an UTN is tested by microscopic simulations in VISSIM. All codes for simulation implementation and solution methods are written in MATLAB.
\subsection{UTN description and Simulation setup}
The tested UTN is built in VISSIM, as shown in Fig. \ref{fig_testednetwork}. It consists of $10$ boundary junctions (marked by blue circles) and $24$ internal junctions (marked by red/yellow squares). There are two roads with different directions between every pair of connecting junctions. The length of roads is in the range of $[280, 400]$ meters. For the junctions corresponding to red (resp. yellow) squares, the sequences of their traffic signal phases are set as Type 3 (resp. Type 2) in TABLE. \ref{tbl_sequencePhases} and their incoming roads have $5$ (resp. 3) lanes. We assume that only one lane is reserved for turning left on every incoming road of junctions with Type 3. The saturation flows of roads are set in the range $[0.55, 0.68]$ (veh./second/lane).
The downstream traffic flow in each destination road has an upper bound of $36$ (veh./min). The ratios for turning right and turning left movements in roads are set randomly in the range $[0.15, 0.25]$.

\begin{figure}[htb]
\centering
\includegraphics[width=0.3\textwidth]{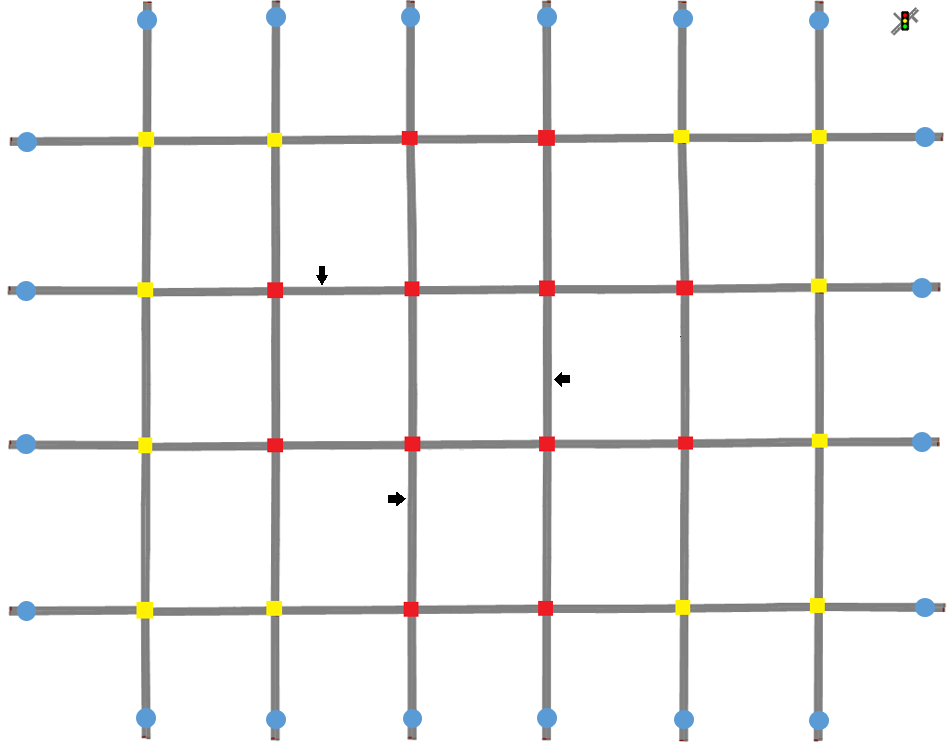}
\caption{The tested UTN built in VISSIM.}\label{fig_testednetwork}
\end{figure}

We set the control time length as $T = 60$ seconds and the lost time as $L = 4$ seconds for all junctions. The total time for each simulation is chosen as $2$ hours corresponding to $120$ control time steps. In this part, three scenarios of traffic demands corresponding to under-saturated, saturated and over-saturated conditions are considered.
\begin{table}
\centering
\caption{Nominal traffic demand $d_B^*(\Delta t)$ (veh./hour).}\label{tb_demand}
\scalebox{0.725}{
\begin{tabular}{c|cccccc}
 \hline
 $\Delta t$ (min.)  & 1 - 20 & 21 - 40 & 41 - 60 & 61 - 80 & 81 - 100 & 101 - 120 \\
 \hline
under-saturated & $1400$ & $1600$ & $1450$ & $1300$ & $1150$ & $1000$\\
saturated & $1500$ & $1800$ & $1650$ & $1500$ & $1350$ & $1200$\\
over-saturated & $1700$ & $2000$ & $1850$ & $1700$ & $1450$ & $1200$\\
 \hline
\end{tabular}}
\end{table}
The traffic demand for each source road having three lanes (resp. five lanes) is set as $d_B^*(\Delta t)$ (resp. $1.25 d_B^*(\Delta t)$), where $d_B^*(\Delta t)$ is described in TABLE. \ref{tb_demand}. In addition to boundary junctions, vehicles enter the UTN from places represented by black arrow in Fig. \ref{fig_testednetwork}. We set these inflows have the rate of $550$ (veh/hour). 

In the following, we compare four traffic control strategies:
\begin{enumerate}
\item Strategy 1 (Predetermined traffic signal control without perimeter control): There are no perimeter control, i.e., the exogenous traffic inflows at source roads are set equal to the corresponding traffic demands, i.e., $f_z^u(t+k|t) = d_z(t+k|t) \forall z: \sigma(z) \in \mathcal{J}^B, k \ge 0$. For control decisions of internal junctions, we run several predetermined signal control time plans and choose the best ones for each traffic demand scenarios. We refer this strategy as the standard control method, which will be used as benchmark to compare with other traffic control strategies.
\item Strategy 2 (Two-layer adaptive signal control framework integrating max pressure with perimeter control): This approach is proposed in \cite{DimitriosTsitsokas2023} combining centralized, aggregated perimeter control strategy, with distributed Max Pressure feedback controllers for internal junctions.
\item Strategy 3 (Weighted MPC-based traffic perimeter and signal control): The control decisions for boundary and internal junctions are determined by solving an MPC-based traffic control problem (\ref{eq_problemTP}-\ref{eq_smooth_UTN}) with the cost function $\Phi(t) = \theta\Phi^{(1)}(t) + 0.25\Phi^{(2)} + \Phi^{(3)}(t)$. The weight $\theta$ corresponds to the aim of maximizing the capacity of the UTN. This weight should be chosen significantly large. In this part, we choose $\theta = 5000$.
\item Strategy 4 (Lexicographic MPC-based traffic perimeter and signal control): This is our proposed traffic control strategy described as the diagram in Fig. \ref{fig_diagram}. For the constraint \eqref{eq_smooth_UTN} aiming to guaranteeing the smooth operation of the UTN, we choose $\gamma_z = 0.5, \forall z \in \mathcal{L}$.
\end{enumerate}
We choose the horizontal time length $K = 4$ for all MPC-based traffic control problems in Strategy 3 and Strategy 4.
\subsection{Results of microscopic simulation}
Fig. \ref{fig_s1results}, Fig. \ref{fig_s2results}, and Fig. \ref{fig_s3results} represent simulation results in under-saturated, saturated, and over-saturated scenarios, respectively. The black, red, green, and blue lines are corresponding to the simulation results of Strategy 1, Strategy 2, Strategy 3, and Strategy 4, respectively.
\begin{figure}[htb]
\centering
\includegraphics[width=0.45\textwidth]{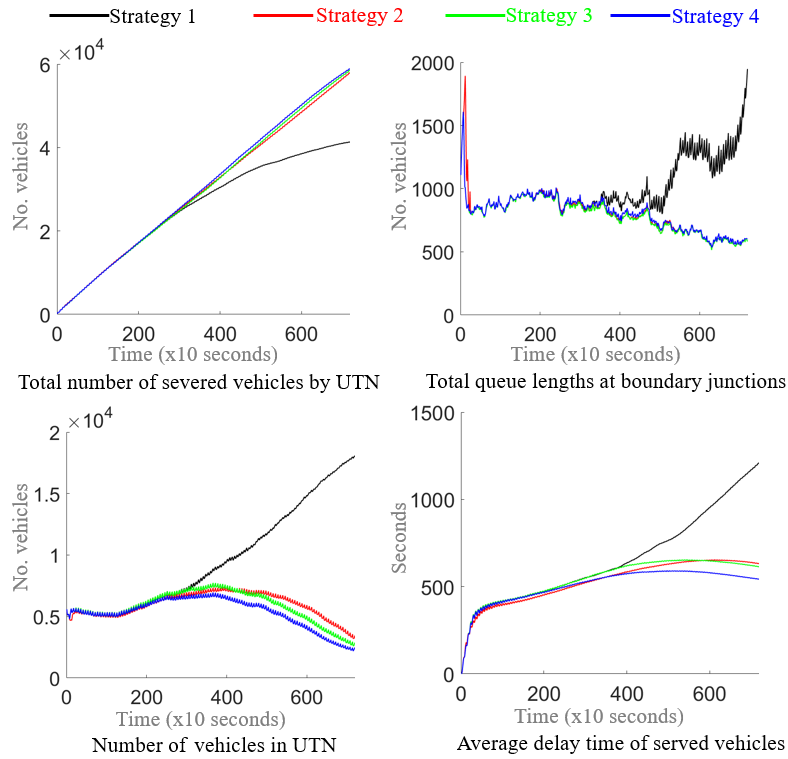}
\caption{Simulation results in under-saturated scenario. }\label{fig_s1results}
\end{figure}
\begin{figure}[htb]
\centering
\includegraphics[width=0.45\textwidth]{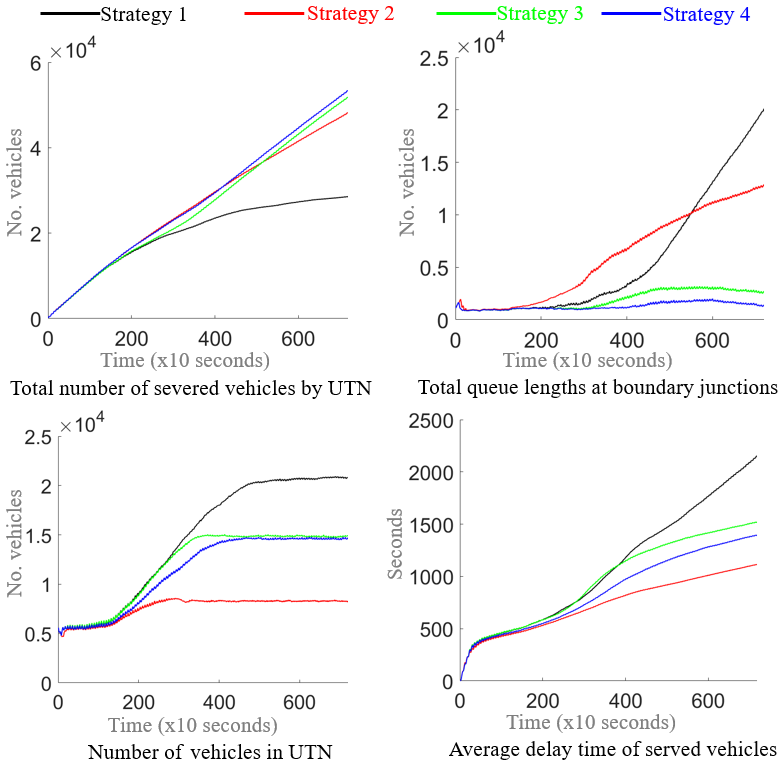}
\caption{Simulation results in saturated scenario.}\label{fig_s2results}
\end{figure}
\begin{figure}[htb]
\centering
\includegraphics[width=0.45\textwidth]{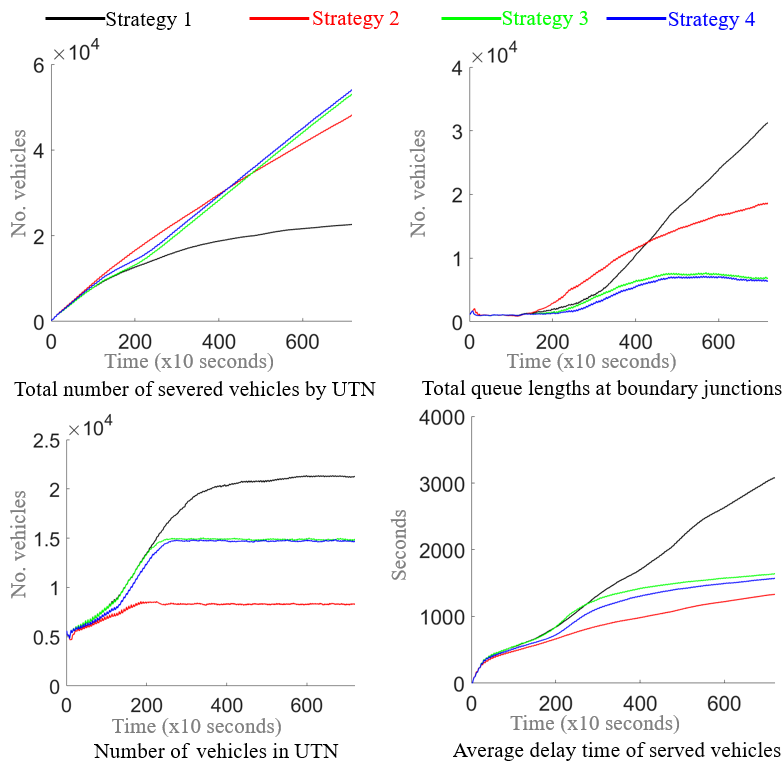}
\caption{Simulation results in over-saturated scenario.}\label{fig_s3results}
\end{figure}
In Fig. \ref{fig_s1results} (similar to Fig. \ref{fig_s2results} and Fig. \ref{fig_s3results}), the top-left figure depicts the number of vehicles served by UTN over the simulation time; the top-right figure is the evolution of total queue lengths at boundary junctions; the bottom-left figure shows the network accumulation; and the bottom-right figures present the average travel time delay of vehicles served by the UTN over the simulation time. The travel time delay of vehicles are measured by VISSIM. Each step in the horizontal axes of these figures is equivalent to ten seconds.

\begin{figure}[htb]
\centering
\includegraphics[width=0.45\textwidth]{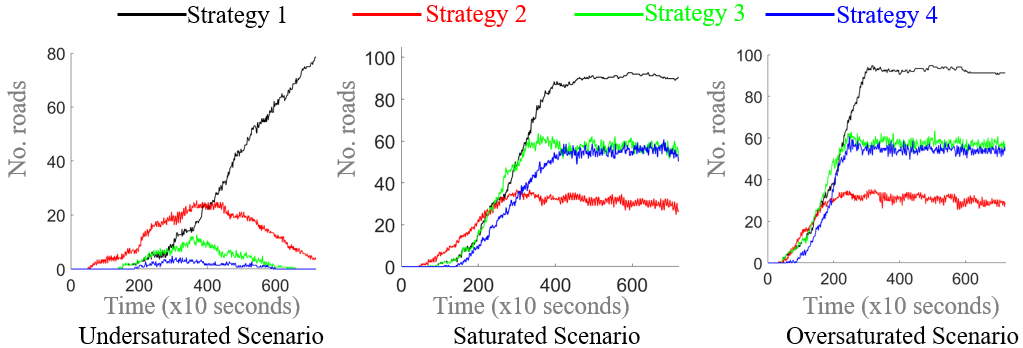}
\caption{No. roads with relative occupation larger than 0.6.}\label{fig_medium_roads}
\end{figure}
\begin{figure}[htb]
\centering
\includegraphics[width=0.45\textwidth]{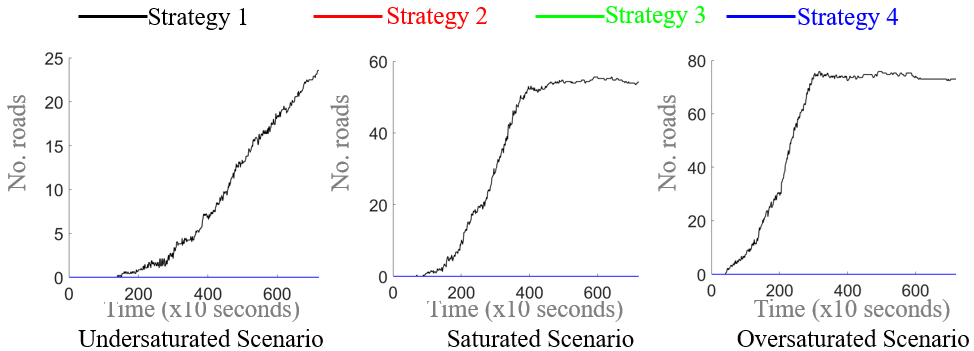}
\caption{No. roads with relative occupation larger than 0.8.}\label{fig_high_roads}
\end{figure}

Another interested performance index is the level of traffic congestion in the UTN, which is measured by the relative occupation of all roads.
Define $n_z^{rel}(\Delta \tilde{t}) = \frac{n_z(10\Delta \tilde{t})}{\overline{n_z}}$ for every road where the index $\Delta \tilde{t} = 1, 2, \dots, 720$.
Let $\mathcal{O}_{r}(\Delta \tilde{t}) = \{z \in \mathcal{R}: n_z^{rel}(\Delta \tilde{t}) > r\}$ be the set of road links having the relative occupation larger than $r$.
Fig. \ref{fig_medium_roads} and Fig. \ref{fig_high_roads} present the evolution of $|\mathcal{O}_{0.6}(\Delta \tilde{t})|$ and $|\mathcal{O}_{0.8}(\Delta \tilde{t})|$, respectively. We can see that the traffic control strategies using perimeter concept (i.e., Strategy 2-4) keep the relative occupation of all roads less than $0.8$ (as shown in Fig. \ref{fig_high_roads}). 

From the performance indexes shown in Fig. \ref{fig_s1results}–\ref{fig_scresults}, it is evident that advanced traffic control strategies (i.e., Strategies 2–4) significantly improve traffic conditions compared to the standard control method (Strategy 1). These advanced strategies not only reduce the travel time delays of vehicles and increase the total network throughput, but also enable the UTN to accommodate a larger number of vehicles. In particular, when traffic demand increases (under saturated and oversaturated scenarios), Strategy 1 faces a high risk of traffic congestion as multiple road links experience high relative occupancy. With Strategy 2, congestion starts to emerge due to the excessive number of vehicles entering the UTN. Consequently, travel time delays rise substantially, and the capacity of the UTN to admit vehicles is noticeably reduced. By restricting traffic inflows into the UTN, Strategies 2–4 are able to prevent the onset of traffic congestion.

In the undersaturated scenario, the performance indexes of Strategies 2–4 are quite similar. Strategy 4 achieves the highest number of served vehicles and the lowest delay; however, the differences are marginal. Under saturated and oversaturated scenarios, Strategy 3 and Strategy 4 can admit more vehicles than Strategy 2 while still preventing traffic congestion within the UTN. Nevertheless, due to the reduction of available free space, the average travel time delays of vehicles operating in the UTN under Strategies 3 and 4 are higher than those under Strategy 2. This outcome is predictable, since the max-pressure control method is known to maximize throughput and minimize travel time delay when traffic demand is low. However, this does not imply that Strategy 2 is superior to Strategies 3 and 4 in enhancing mobility, as it prolongs vehicle waiting times at boundary junctions. Owing to their ability to accommodate a larger number of vehicles, the MPC-based control methods (Strategies 3 and 4) ultimately serve more vehicles overall.
In addition, we can reduce the average travel time delay of served vehicles for Strategy 4 (similar in Strategy 3) by decreasing the parameter $\gamma_z$ in the constraint \eqref{eq_smooth_UTN}. Let Strategy 5 be similar to Strategy 4, except setting the parameter $\gamma_z = 0.3 \forall z \in \mathcal{L}$. Fig. \ref{fig_scresults} shows the simulation results in the saturated scenario for Strategy 1 (pink lines), Strategy 4 (red lines), and Strategy 5 (black lines).
By reducing the parameter $\gamma_z, \forall z \in \mathcal{L}$, Strategy 5 forces each road link to reserve more free space for vehicles. We can see that the average travel time delay is significantly reduced when using Strategy 5, but, in return, the total queue length at boundary junctions increases (compared to Strategy 4).
\begin{figure}[htb]
\centering
\includegraphics[width=0.45\textwidth]{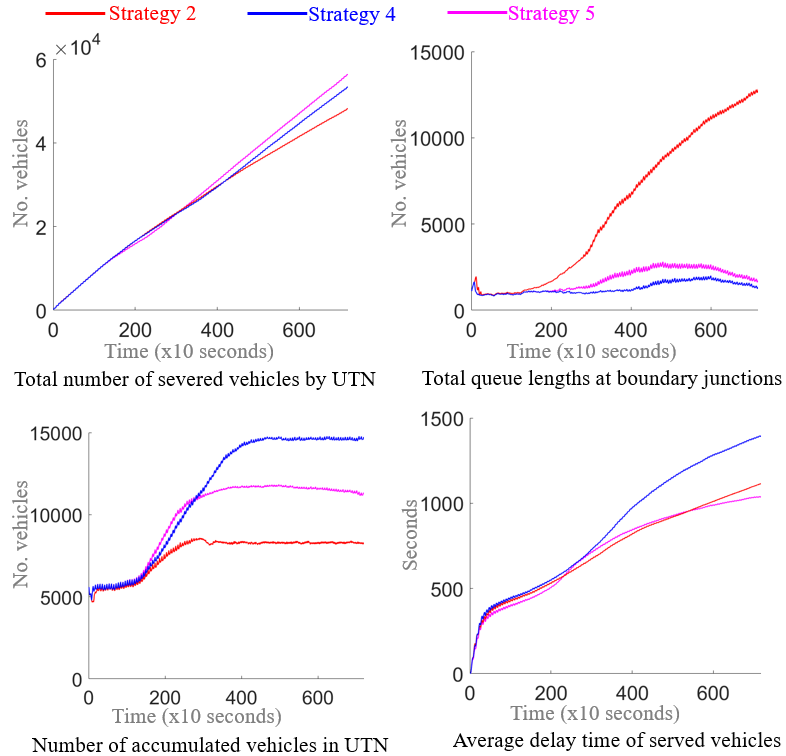}
\caption{Simulation results in saturated scenario for Strategy 1, Strategy 4 and Strategy 5.}\label{fig_scresults}
\end{figure}

Based on the simulation results shown in Fig. \ref{fig_s1results}-\ref{fig_scresults}, we can conclude that our proposed traffic control strategy successfully integrates perimeter control and traffic signal control. It not only protects the UTN from the risk of traffic congestion but also optimizes the capacity of the UTN.
Strategy 3 using a weighted combination of multiple objectives can approach more closer to the control performance of Strategy 4 if setting the weight $\theta$ larger. However, this setting may also make a bigger computation load as shown in next subsection.
\subsection{Computational load of the proposed control method}
In this paper, Algorithm \ref{alg_proposed_opt} is implemented for Strategy 3 and Strategy 4 by MATLAB on a computer having chip Intel Core i5 8500 and $16$ Gb RAM. The penalty parameter is chosen as $\rho = 1 $ for the linear program (i.e., the problem $\textbf{(PC)}$ in Strategy 4) and $\rho = 0.1$ for quadratic programs (i.e., the problem $\textbf{(TSC)}$ in Strategy 4 and the problem (\ref{eq_problemTP}-\ref{eq_smooth_UTN}) in Strategy 3). The convergence of Algorithm \ref{alg_proposed_opt} for solving MPC-based traffic control problems is illustrated in Fig. \ref{fig_computation_typical}.
\begin{figure}[htb]
\centering
\includegraphics[width=0.48\textwidth]{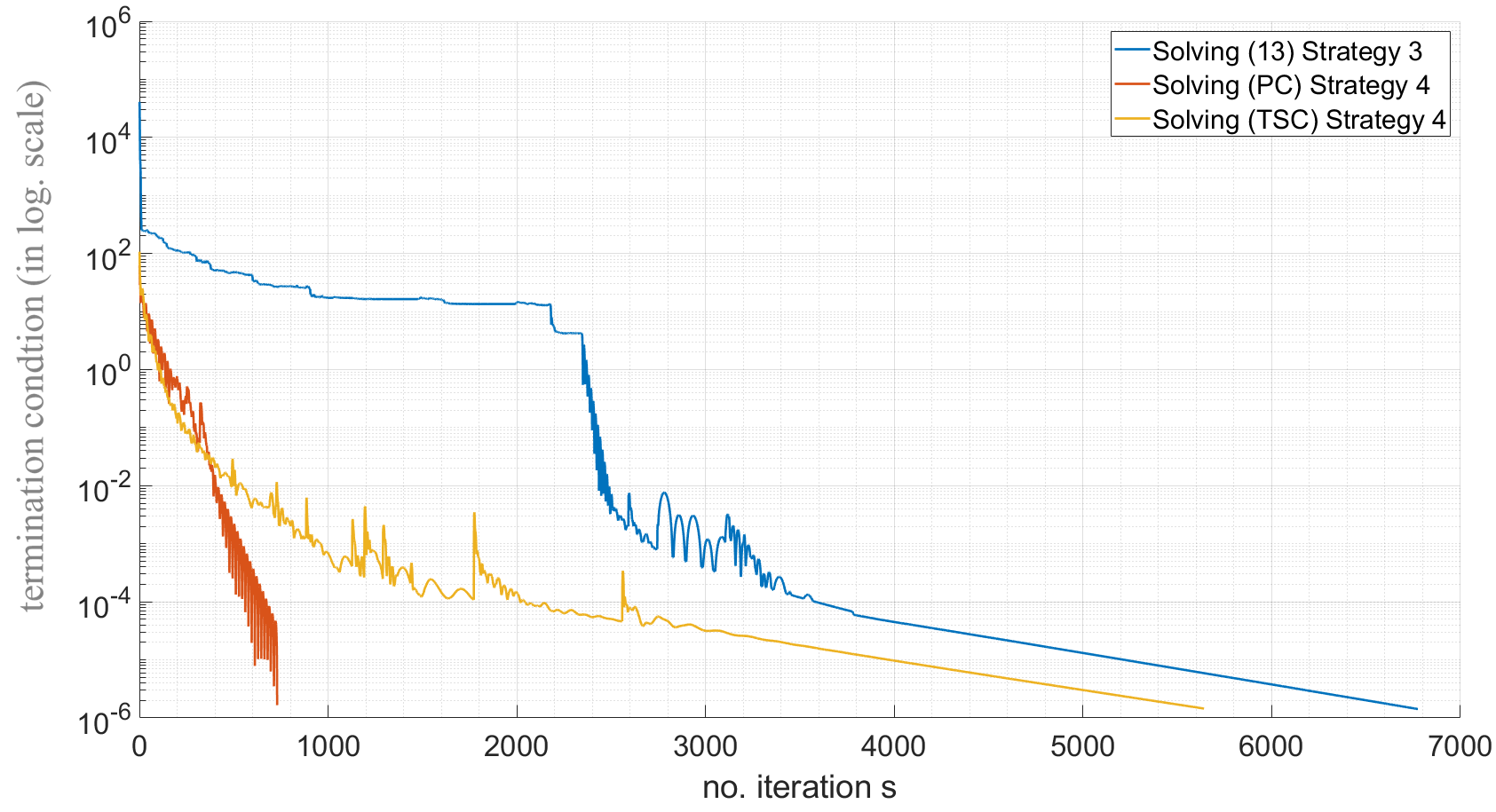}
\includegraphics[width=0.48\textwidth]{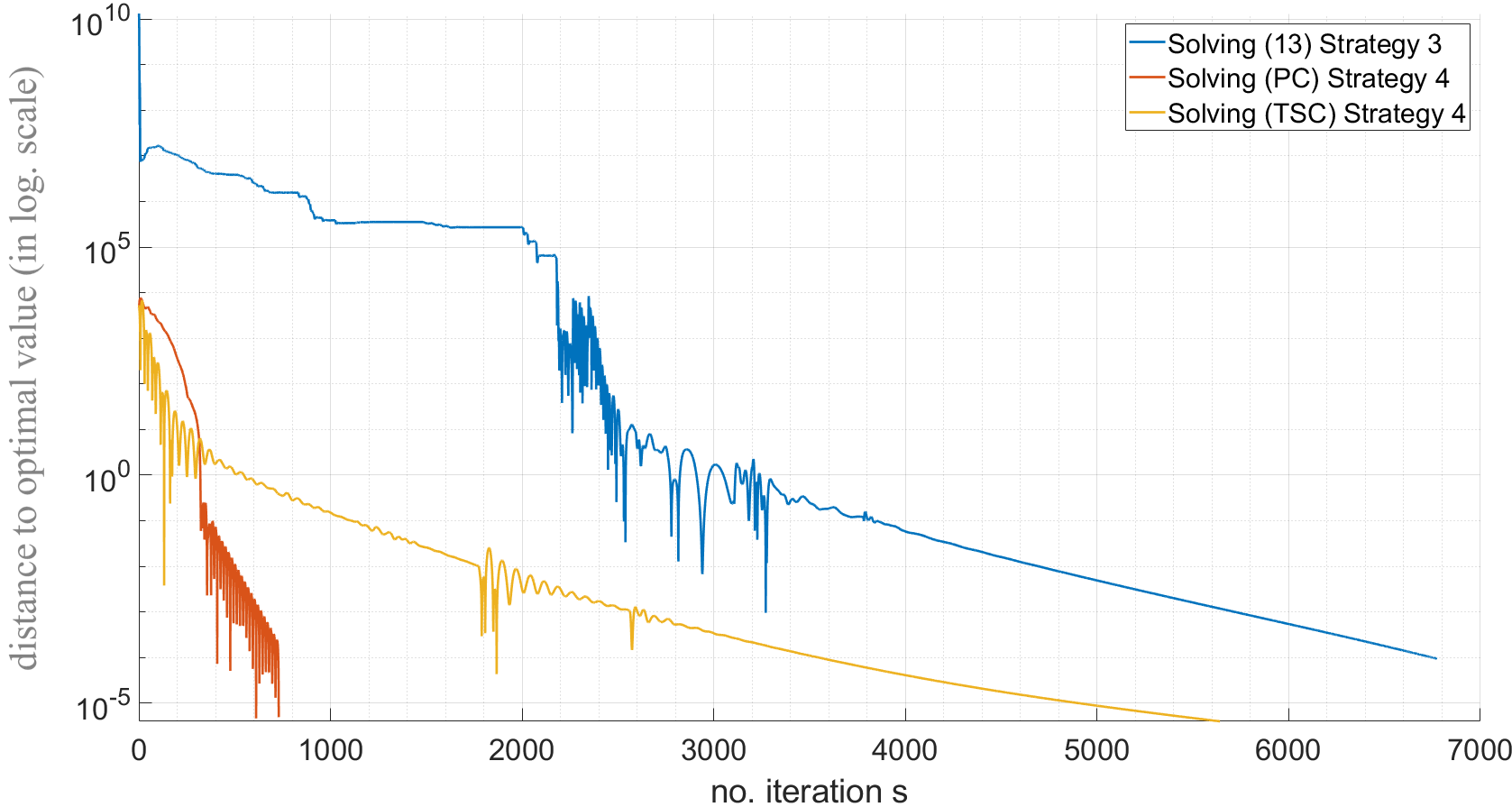}
\caption{The convergence of Algorithm \ref{alg_proposed_opt} when solving MPC-based traffic control problems. The top figure shows the evolution of termination conditions, while the bottom figure present the difference between the estimated cost and the optimal cost value.}\label{fig_computation_typical}
\end{figure}
The top figure shows the evolution of termination condition, which is formulated as $\sum_{i = 1}^{N}\Bigl(||\textbf{V}_i\textbf{x}_i(s) - \textbf{y}_i(s) - \textbf{v}_i||_{\infty} + \sum_{\mathcal{S}_j \in \mathcal{N}_{\mathcal{S}_i}}||\textbf{U}_{ij}\textbf{x}_i(s) - \textbf{y}_{ij}(s)||_{\infty}\Bigr)$. The bottom figure is corresponding to the convergence of the estimated cost to the optimal cost value.
Based on the simulation results in these figures, it is reasonable to choose $tol = 10^{-5}/\rho$ for the terminated condition \eqref{eq_terminated}.
We also observe that the number of iterations required for the problem $\textbf{(PC)}$ in Strategy 4 is significant less than the one required for two quadratic programs. Because of large parameter $\theta$ in the cost function, it takes more time to solve the MPC-based traffic control problem in Strategy 3 than the problem $\textbf{(TSC)}$ in Strategy 4.

We further measure the time execution when employing Algorithm \ref{alg_proposed_opt} in distributed manner. Let the considered UTN be divided into $24$ subregions. Each of them consists of one internal junction and the connecting boundary junctions (if existed). For the distributed setup, we use functions \textit{tic} and \textit{toc} in MATLAB to measure the time taken by one local controller to update its local variables in every iteration of Algorithm \ref{alg_proposed_opt}. The maximum value of the time length taken by local controllers is added to the total execution time for running Algorithm \ref{alg_proposed_opt}. This process is repeated until the termination condition \eqref{eq_terminated} is satisfied.
\begin{figure}[htb]
\centering
\includegraphics[width=0.35\textwidth]{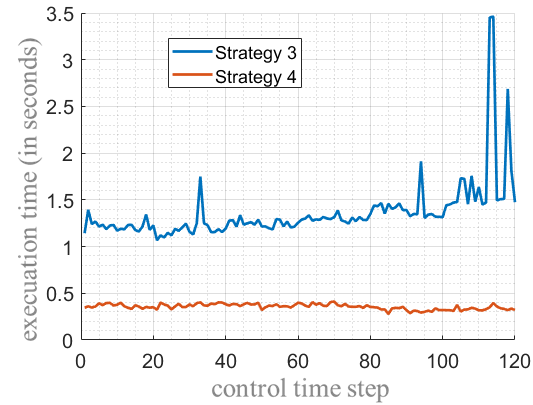}
\includegraphics[width=0.35\textwidth]{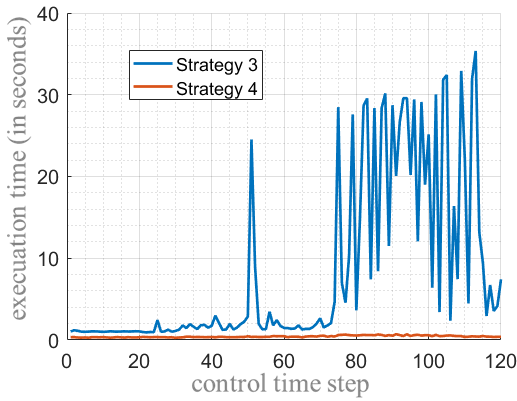}
\caption{The required execution time in each cycle for under-saturated (top) and saturated (bottom) scenarios.}\label{fig_computation_strategies}
\end{figure}
The execution time when using Algorithm \ref{alg_proposed_opt} to determine the control decisions in each control time step of Strategy 3 and Strategy 4 is presented in Fig. \ref{fig_computation_strategies}. The top figure is corresponding to the under-saturated scenario. The bottom figure shows simulation results for the saturated scenario, which is similar to the one for over-saturated scenario. Note that, the execution time of Strategy 4 is the sum of the ones in solving two MPC-based traffic control problems, $\textbf{(PC)}$ and $\textbf{(TSC)}$.
TABLE~\ref{tb_convergence} shows the average and maximum computation load when using Algorithm \ref{alg_proposed_opt} to solve the problems $\textbf{(PC)}$ and $\textbf{(TSC)}$ in distributed manner for some different horizontal time $K$. In each cell of this table, we give the number of required iterations in black and the execution time (seconds) in blue.
\begin{table}[htb]
\centering
\caption{Computational load of Algorithm \ref{alg_proposed_opt} in Strategy 4 with different $K$.}
\label{tb_convergence}
\scalebox{0.85}{
\begin{tabular}{c|c|c|c|c}
\hline
\multirow{2}{*}{K} & \multicolumn{2}{c|}{Solving $\textbf{(PC)}$} & \multicolumn{2}{c}{Solving $\textbf{(TSC)}$}\\
\cline{2-5}
 & average & maximum & average & maximum\\
\hline
1 & $354$ \textcolor{blue}{$(0.025)$} & $420$ \textcolor{blue}{$(0.031)$} & $744$ \textcolor{blue}{$(0.051)$} & $924$ \textcolor{blue}{$(0.0643)$}\\
\hline
2 & $504$ \textcolor{blue}{$(0.048)$} & $738$ \textcolor{blue}{$(0.0712)$} & $980$ \textcolor{blue}{$(0.114)$} & $1052$ \textcolor{blue}{$(0.132)$}\\
\hline
3 & $541$ \textcolor{blue}{$(0.098)$} & $817$ \textcolor{blue}{$(0.148)$} & $1225$ \textcolor{blue}{$(0.285)$} & $1635$ \textcolor{blue}{$(0.382)$}\\
\hline
4 & $595$ \textcolor{blue}{$(0.12)$} & $858$ \textcolor{blue}{$(0.173)$} & $1458$ \textcolor{blue}{$(0.32)$} & $1867$ \textcolor{blue}{$(0.419)$}\\
 \hline
5 & $827$ \textcolor{blue}{$(0.361)$} & $977$ \textcolor{blue}{$(0.427)$} & $1596$ \textcolor{blue}{$(1.162)$} & $1988$ \textcolor{blue}{$1.448()$}\\
 \hline
6 & $1006$ \textcolor{blue}{$(0.507)$} & $1230$ \textcolor{blue}{$(0.619)$} & $1951$ \textcolor{blue}{$(1.374)$} & $2155$ \textcolor{blue}{$(1.527)$}\\
 \hline
\end{tabular}}
\end{table}
As the execution time is significantly shorter than the time interval between two consecutive control steps in all cases, it is reasonable to state that our proposed traffic control strategy can be used in real-time when good communication among local controllers is available.
\section{Conclusion}
Traffic signal control and perimeter control are widely adopted strategies for alleviating congestion in UTNs, but each has inherent limitations. The primary objective of traffic signal control is to optimize existing road infrastructure for increasing network throughput and reducing vehicle delays. However, even the most advanced signal control systems fall short of preventing congestion when the volume of incoming traffic exceeds the network's capacity. Current research that incorporates perimeter control, whether used alone or in combination with signal control, often fails to produce a fully optimized signal control plan that maximizes overall network efficiency.
Motivated by these observations, this paper presents a novel framework that integrates traffic perimeter control with traffic signal control, formulated as a lexicographic multi-objective optimization problem. The proposed approach initially regulates traffic inflows at boundary junctions to maximize network capacity while ensuring a smooth operation of the whole UTN. Following this, the signal timings are collaboratively optimized to improve traffic conditions under the regulated inflows. An MPC strategy is employed to structure the control problems, ensuring adherence to safety and capacity constraints at junctions and on roads.
To manage the computational complexity, we decompose the UTN into multiple subnetworks, each managed by a local agent. A distributed solution method based on ADMM is developed, enabling each agent to optimize control decisions using only local information from its own subnetwork and neighboring agents.
The effectiveness of the proposed framework is validated through simulations using VISSIM and MATLAB, demonstrating significant improvements in traffic network management.
\appendix
\subsection{Distributed solution method based on ADMM}\label{subADMM}
In this part, we use an Alternating Direction Multiplier Method (ADMM) algorithm to design a distributed solution method for solving the following optimization problem:
\begin{subequations}\label{eq_distributedOpt}
\begin{align}
\min\limits_{\textbf{x}_i, \forall i} \textrm{ }& \sum\limits_{i = 1}^{N}\left(\frac{1}{2}\textbf{x}_i^T\textbf{W}_i\textbf{x}_i + \textbf{w}_i^T\textbf{x}_i\right)\\
\textrm{s.t. }& \textbf{U}_i\textbf{x}_i = \textbf{u}_i, \forall i = 1, \dots, N,\\
& \textbf{V}_i\textbf{x}_i \le \textbf{v}_i, \forall i = 1, \dots, N,\\
& \textbf{U}_{ij}\textbf{x}_i + \textbf{U}_{ji}\textbf{x}_j = \textbf{0}, \forall \mathcal{S}_j \in \mathcal{N}_{\mathcal{S}_i}, i = 1, \dots, N.
\end{align}
\end{subequations}
where $\textbf{x}_i$ is a variable vector, $\textbf{U}_i, \textbf{V}_i$, $\textbf{U}_{ij} \forall \mathcal{S}_j \in \mathcal{N}_{\mathcal{S}_i}$, and $\textbf{u}_i, \textbf{v}_i$ are constant matrices and vectors with suitable dimension known by only the agent $\mathcal{S}_i$, for all $i = 1, \dots, N$. Here, we assume that the matrix $\textbf{U}_i$ is full row rank and the matrix $\textbf{W}_i$ is positive semidefinite, for all $i = 1, \dots, N$.
Let $\textbf{y}_i = \textbf{V}_i\textbf{x}_i - \textbf{v}_i$, $\textbf{y}_{ij} = \textbf{U}_{ij}\textbf{x}_i \forall \mathcal{S}_j \in \mathcal{N}_{\mathcal{S}_i}$ and define the sets $\mathcal{X}_i = \{\textbf{x}_i: \textbf{U}_i\textbf{x}_i = \textbf{u}_i\}, \mathcal{Y}_i = \{\textbf{y}_i: \textbf{y}_i \le \textbf{0}\}$, $\Omega_{ij} = \{(\textbf{y}_{ij}, \textbf{y}_{ji}): \textbf{y}_{ij} + \textbf{y}_{ji} = \textbf{0}\}$. The problem \eqref{eq_distributedOpt} can be rewritten as
\begin{subequations}\label{eq_distributedOpt_ADMM}
\begin{align}
\min\limits_{\substack{\textbf{x}_i \in \mathcal{X}_i, \textbf{y}_i \in \mathcal{Y}_i,\\ (\textbf{y}_{ij}, \textbf{y}_{ji}) \in \Omega_{ij} \forall \mathcal{S}_j \in \mathcal{N}_{\mathcal{S}_i}}}&  \sum\limits_{i = 1}^{N}\left(\frac{1}{2}\textbf{x}_i^T\textbf{W}_i\textbf{x}_i + \textbf{w}_i^T\textbf{x}_i\right)\\
\textrm{ s.t. }&\left\{ \begin{matrix} \textbf{V}_i\textbf{x}_i - \textbf{v}_i = \textbf{b}_i,\\ \textbf{U}_{ij}\textbf{x}_i - \textbf{y}_{ij} = \textbf{0}, \forall \mathcal{S}_j \in \mathcal{N}_{\mathcal{S}_i},\\ 
\end{matrix}\right., \forall i = 1, \dots, N.
\end{align}
\end{subequations}
It is easy to verify that the problem \eqref{eq_distributedOpt_ADMM} has the form of \eqref{eq_ADMMproblem}.
\begin{equation}\label{eq_ADMMproblem}
\min\limits_{\textbf{x} \in \mathcal{X}, \textbf{y} \in \mathcal{Y}} \textrm{ }\Psi_x(\textbf{x}) + \Psi_y(\textbf{y}) \textrm{ s.t. } \textbf{X}\textbf{x} + \textbf{Y}\textbf{y} = \textbf{z}.
\end{equation}
where $\textbf{x} = \textrm{col}\{\textbf{x}_1, \textbf{x}_2, \cdots, \textbf{x}_N\}$, $\textbf{y} = \textrm{col}\{\textrm{col}\{\textbf{y}_i, \textrm{col}\{\textbf{y}_{ij}: \mathcal{S}_j \in \mathcal{N}_{\mathcal{S}_i}\}\}: 1 \le i \le N\}$, $\Psi_x(\textbf{x}) = \sum_{i = 1}^{N}\left(\frac{1}{2}\textbf{x}_i^T\textbf{W}_i\textbf{x}_i + \textbf{w}_i^T\textbf{x}_i\right), \Psi_y(\textbf{y}) = 0$, $\textbf{X} = \textrm{blkdiag}\{\textbf{X}_1, \textbf{X}_2, \dots, \textbf{X}_N\}$  where $\textbf{X}_i = \textrm{blkdiag}\{\textbf{V}_i, \textrm{blkdiag}\{\textbf{U}_{ij}: \mathcal{S}_j \in \mathcal{N}_{\mathcal{S}_i}\}\}$, $\textbf{Y} = -\textbf{I}$, $\textbf{z} = \textrm{col}\{\textbf{z}_1, \textbf{z}_2, \dots, \textbf{z}_N\}$ where $\textbf{z}_i = \textrm{col}\{\textbf{v}_i, \textbf{0}\}$, and the sets $\mathcal{X} = \bigtimes_{i = 1}^{N} \mathcal{X}_i$, $\mathcal{Y} = \bigtimes_{i = 1}^{N} \left( \Omega_i \times \bigtimes_{\mathcal{S}_j \in \mathcal{N}_{\mathcal{S}_i}} \Omega_{ij}\right)$. Here, $\times$ denotes the Cartesian product of two sets and $\bigtimes_{i = 1}^{N} \mathcal{X}_i = \mathcal{X}_1 \times \cdots \times \mathcal{X}_N$.
\subsubsection{ADMM algorithm}
The ADMM algorithm proposed in \cite{BingshengHe2015} for solving \eqref{eq_ADMMproblem} is given by the following iteration update:
\begin{subequations}\label{eq_ADMM}
\begin{align}
\textbf{x}(s+1) &= \arg\min\limits_{\textbf{x} \in \mathcal{X}} \left\{\mathcal{L}(\textbf{x},\textbf{y}(s),\boldsymbol{\lambda}(s)) + \frac{1}{2}\left|\left|\textbf{x} - \textbf{x}(s)\right|\right|_{\textbf{G}}^2\right\},\\
\textbf{y}(s+1) &= \arg\min\limits_{\textbf{y} \in \mathcal{Y}} \mathcal{L}(\textbf{x}(s+1),\textbf{y},\boldsymbol{\lambda}(s)),\\
\boldsymbol{\lambda}(s+1) &= \boldsymbol{\lambda}(s) - \rho \Bigl( \textbf{X}\textbf{x}(s+1) + \textbf{Y}\textbf{y}(s+1) - \textbf{z} \Bigr).
\end{align}
\end{subequations}
where $\mathcal{L}(\textbf{x},\textbf{y},\boldsymbol{\lambda})$ is the augmented Lagrangian function defined by $\mathcal{L}(\textbf{x},\textbf{y},\boldsymbol{\lambda}) = \Psi_x(\textbf{x}) + \Psi_y(\textbf{y}) + \frac{\rho}{2}\left|\left|\textbf{X}\textbf{x} + \textbf{Y}\textbf{y} - \textbf{z} - \frac{1}{\rho}\boldsymbol{\lambda}\right|\right|^2$, $\boldsymbol{\lambda}$ is the dual variable associated with the equality constraint in \eqref{eq_ADMMproblem}, $\rho > 0$ is a penalty parameter, and $\textbf{G}$ is a symmetric and positive semidefinite matrix.
In \cite{BingshengHe2015, WeiDeng2016, XiaoweiPan2022}, the convergence rate of the ADMM algorithm \eqref{eq_ADMM} is proved to be sublinear for arbitrarily chosen positive $\rho$. We summarize the convergence results in the following theorem:
\begin{Theorem}\label{th_ADMM}
Let $\textbf{G}$ be a positive definite matrix. Then, $\lim\limits_{s \rightarrow \infty}||\textbf{x}(s) - \textbf{x}^{opt}|| = 0$ where $\textbf{x}^{opt} = \textrm{col}\{\textbf{x}_i^{opt}: 1 \le i \le N\}$ is one optimal solution of the problem \eqref{eq_distributedOpt}. In addition, \[\left|\sum_{i = 1}^{N}\left(\frac{1}{2}\textbf{x}_i(s)^T\textbf{W}_i\textbf{x}_i(s) + \textbf{w}_i^T\textbf{x}_i(s)\right) - \Psi^{opt}\right| = o\left(\frac{1}{\sqrt{s}}\right).\]
where $\Psi^{opt}$ is the optimal cost value of the problem \eqref{eq_distributedOpt}
\end{Theorem}
\subsubsection{Detailed update}
Let $\boldsymbol{\lambda}_i$ and $\boldsymbol{\lambda}_{ij} \forall \mathcal{S}_j \in \mathcal{N}_{\mathcal{S}_i}, \forall i = 1, 2, \dots, N$, be the dual variables corresponding to equality constraints in \eqref{eq_distributedOpt_ADMM}, respectively. We have the detailed form of the augmented Lagrangian function as $\mathcal{L}(\textbf{x},\textbf{y},\boldsymbol{\lambda}) = \sum_{i = 1}^{N} \mathcal{L}_i(\textbf{x}_i,\tilde{\textbf{y}}_i,\tilde{\boldsymbol{\lambda}}_i)$ where $\tilde{\textbf{y}}_i = \textrm{col}\{\textbf{y}_i, \textrm{col}\{\textbf{y}_{ij}: \mathcal{S}_j \in \mathcal{N}_{\mathcal{S}_i}\}\}$, $\tilde{\boldsymbol{\lambda}}_i = \textrm{col}\{\boldsymbol{\lambda}_i, \textrm{col}\{\boldsymbol{\lambda}_{ij}: \mathcal{S}_j \in \mathcal{N}_{\mathcal{S}_i}\}\}$ and $\mathcal{L}_i(\textbf{x}_i,\tilde{\textbf{y}}_i,\tilde{\boldsymbol{\lambda}}_i) = \frac{1}{2}\textbf{x}_i^T\textbf{W}_i\textbf{x}_i + \textbf{w}_i^T\textbf{x}_i + \frac{\rho}{2}\left(\left|\left|\textbf{V}_i\textbf{x}_i - \textbf{y}_i - \textbf{v}_i - \frac{1}{\rho}\boldsymbol{\lambda}_i\right|\right|^2 + \sum\limits_{\mathcal{S}_j \in \mathcal{N}_{\mathcal{S}_i}} \left|\left|\textbf{U}_{ij}\textbf{x}_i - \textbf{y}_{ij} - \frac{1}{\rho}\boldsymbol{\lambda}_{ij}\right|\right|^2\right)$.
In this paper, we choose the matrix $\textbf{G} = \textrm{blkdiag}\{\textbf{G}_i: 1 \le i \le N\}$ where $\textbf{G}_i$ is a positive definite matrix.
Due to the separation of the augmented Lagrangian function as $\mathcal{L}(\textbf{x},\textbf{y},\boldsymbol{\lambda})$ and the chosen matrix $\textbf{G}$, the equations (\ref{eq_ADMM}a) and (\ref{eq_ADMM}b) are equivalent to the equations \eqref{eq_updatelaw_primal_1} and \eqref{eq_updatelaw_primal_2}, $\forall i = 1, 2, \dots, N$, respectively.
\begin{align}
\textbf{x}_i(s+1) = \arg\min\limits_{\textbf{U}_i\textbf{x}_i = \textbf{u}_i} &\Biggl\{\frac{1}{2}\textbf{x}_i^T\textbf{W}_i\textbf{x}_i + \textbf{w}_i^T\textbf{x}_i + \frac{\rho}{2}\left(\left|\left|\textbf{V}_i\textbf{x}_i - \textbf{y}_i - \textbf{v}_i - \frac{1}{\rho}\boldsymbol{\lambda}_i\right|\right|^2 + \sum\limits_{\mathcal{S}_j \in \mathcal{N}_{\mathcal{S}_i}} \left|\left|\textbf{U}_{ij}\textbf{x}_i - \textbf{y}_{ij} - \frac{1}{\rho}\boldsymbol{\lambda}_{ij}\right|\right|^2\right)\nonumber\\ &+ \frac{1}{2}\Bigl|\Bigl|\textbf{x}_i - \textbf{x}_i(s)\Bigr|\Bigr|_{\textbf{G}_i}\Biggr\}.\label{eq_updatelaw_primal_1}
\end{align}
\begin{subequations}\label{eq_updatelaw_primal_2}
\begin{align}
\textbf{y}_i(s+1) &= \arg\min\limits_{\textbf{y}_i \le \textbf{0}} \left|\left|\textbf{V}_i\textbf{x}_i(s+1) - \textbf{y}_i - \textbf{v}_i - \frac{1}{\rho}\boldsymbol{\lambda}_i(s)\right|\right|^2,\\
(\textbf{y}_{ij}, \textbf{y}_{ji})(s+1) &= \arg\min\limits_{\textbf{y}_{ij} = \textbf{y}_{ji}} \left(\left|\left|\textbf{U}_{ij}\textbf{x}_i(s+1) - \textbf{y}_{ij} - \frac{1}{\rho}\boldsymbol{\lambda}_{ij}(s)\right|\right|^2 + \left|\left|\textbf{U}_{ji}\textbf{x}_j(s+1) - \textbf{y}_{ji} - \frac{1}{\rho}\boldsymbol{\lambda}_{ji}(s)\right|\right|^2\right), \forall \mathcal{S}_j \in \mathcal{N}_{\mathcal{S}_i}.
\end{align}
\end{subequations}

Define the matrix $\tilde{\textbf{W}}_i$ and the vector $\tilde{\textbf{w}}_i(s)$ as in (\ref{eq_updatelaw_detailed}a) for all $i = 1, 2, \dots, N$. The KKT conditions for the optimization problem in \eqref{eq_updatelaw_primal_1} are given as follows.
\begin{subequations}\label{eq_tempKKTconditions}
\begin{align}
\tilde{\textbf{W}}_i\textbf{x}_i(s+1) - \tilde{\textbf{w}}_i(s) + \textbf{U}_i^T\boldsymbol{\mu}_i &= \textbf{0},\\
\textbf{U}_i\textbf{x}_i &= \textbf{u}_i.
\end{align}
\end{subequations}
where $\boldsymbol{\mu}_i$ is the dual variable corresponding to the equality constraint $\textbf{U}_i\textbf{x}_i = \textbf{u}_i$.
By solving the linear equation \eqref{eq_tempKKTconditions}, we obtain $\textbf{x}_i(s+1) = \tilde{\textbf{W}}_i^{-1}\left(\tilde{\textbf{w}}_i(s) - \textbf{U}_i^T\boldsymbol{\mu}_i(s)\right)$ and $\boldsymbol{\mu}_i(s) = \left(\textbf{U}_i\tilde{\textbf{W}}_i^{-1}\textbf{U}_i^T\right)^{-1}\left(\textbf{U}_i\tilde{\textbf{W}}_i^{-1}\tilde{\textbf{w}}_i(s) - \textbf{u}_i\right)$. Then we have the update for $\textbf{x}_i(s+1)$ as in (\ref{eq_updatelaw_detailed}a).
Using similar analysis, the optimization problem in (\ref{eq_updatelaw_primal_2}b) has the optimal solution $\textbf{y}_{ij}(s+1) = \textbf{y}_{ji}(s+1)$ given in (\ref{eq_updatelaw_detailed}d).
Consider (\ref{eq_updatelaw_primal_2}a), we have $\textbf{y}_i(s+1)$ is the projection of the point $\left(\textbf{V}_i\textbf{x}_i(s+1) - \textbf{v}_i - \frac{1}{\rho}\boldsymbol{\lambda}_i(s)\right)$ onto the set $\mathcal{Y}_i = \{\textbf{y}_i: \textbf{y}_i \le \textbf{0}\}$. So, $\textbf{y}_i(s+1)$ given in (\ref{eq_updatelaw_detailed}b) is the optimal solution for the problem in (\ref{eq_updatelaw_primal_2}a). 
Derived from (\ref{eq_ADMM}c), the detailed equations for updating dual variables $\boldsymbol{\lambda}_i$ and $\boldsymbol{\lambda}_{ij}(s), \forall \mathcal{S}_j \in \mathcal{N}_{\mathcal{S}_i},$ are given by (\ref{eq_updatelaw_detailed}c) and (\ref{eq_updatelaw_detailed}e), respectively.
\begin{subequations}\label{eq_updatelaw_detailed}
\begin{align}
\hat{\textbf{x}}_i(s+1) &= \tilde{\textbf{W}}_i^{-1}\left(\textbf{w}_i(s) - \textbf{U}_i^T\left(\textbf{U}_i\tilde{\textbf{W}}_i^{-1}\textbf{U}_i^T\right)^{-1}\left(\textbf{U}_i\tilde{\textbf{W}}_i^{-1}\tilde{\textbf{w}}_i(s) - \textbf{u}_i\right)\right),\\
\textrm{where }\tilde{\textbf{W}}_i &= \textbf{G}_i + \rho\Bigl(\textbf{V}_i^T\textbf{V}_i + \sum_{\mathcal{S}_j \in \mathcal{N}_{\mathcal{S}_i}}\textbf{U}_{ij}^T\textbf{U}_{ji}\Bigr),\nonumber\\
\tilde{\textbf{w}}_i(s) &= -\textbf{w}_i + \textbf{G}_i\textbf{x}_i(s) - \textbf{V}_i^T\left(\rho\textbf{y}_i + \textbf{v}_i(s) + \boldsymbol{\lambda}_i(s)\right) + \sum_{\mathcal{S}_j \in \mathcal{N}_{\mathcal{S}_i}}\textbf{U}_{ij}^T\left(\rho\textbf{y}_{ij}(s) + \boldsymbol{\lambda}_{ij}(s)\right),\nonumber\\
\textbf{y}_i(s+1) &= \min\left\{\textbf{0}, \textbf{V}_i\textbf{x}_i(s+1) - \textbf{v}_i - \frac{1}{\rho}\boldsymbol{\lambda}_i(s)\right\},\\
\boldsymbol{\lambda}_i(s+1) &= \boldsymbol{\lambda}_i(s) - \rho\left(\textbf{V}_i\textbf{x}_i(s+1) - \textbf{y}_i(s+1) - \textbf{v}_i\right),\\
\textbf{y}_{ij}(s+1) &= \frac{1}{2}\left(\textbf{U}_{ij}\textbf{x}_i(s+1) - \frac{1}{\rho}\boldsymbol{\lambda}_{ij}(s) + \textbf{U}_{ji}\textbf{x}_j(s+1) - \frac{1}{\rho}\boldsymbol{\lambda}_{ji}(s)\right), \forall \mathcal{S}_j \in \mathcal{N}_{\mathcal{S}_i},\\
\boldsymbol{\lambda}_{ij}(s+1) &= \boldsymbol{\lambda}_{ij}(s) - \rho\left(\textbf{U}_{ij}\textbf{x}_i(s+1) - \textbf{y}_{ij}(s+1)\right), \forall \mathcal{S}_j \in \mathcal{N}_{\mathcal{S}_i}.
\end{align}
\end{subequations}
\subsubsection{ADMM-based distributed algorithm}
According Theorem \ref{th_ADMM}, the convergence of the ADMM-based update law \eqref{eq_updatelaw_detailed} is asymptotic. To implement this update method in real-time application, a stopping criteria is necessary.
In this paper, we use the equation \eqref{eq_terminated}, $\forall i = 1, 2, \dots, N$, as the stopping condition and use the min-consensus law \eqref{eq_minconsensus} to verify this condition in distributed manner.
\begin{subequations}\label{eq_terminated}
\begin{gather}
||\textbf{V}_i\textbf{x}_i(s) - \textbf{y}_i(s) - \textbf{v}_i||_{\infty} \le tol,\\
||\textbf{U}_{ij}\textbf{x}_i(s) - \textbf{y}_{ij}(s)||_{\infty} \le tol, \forall j: \mathcal{S}_i \in \mathcal{N}_j,
\end{gather}
\end{subequations}
where $tol$ is a given small positive tolerance.
\begin{equation}\label{eq_minconsensus}
fl_{i,s}(\varsigma+1) = \min\left\{fl_{j,s}(\varsigma): \mathcal{S}_j \in \mathcal{N}_{\mathcal{S}_i} \cup \{\mathcal{S}_i\}\right\}
\end{equation}
where $fl_{i,s}$ is a a flag of the agent $\mathcal{S}_i$ corresponding to the iteration $s$ of the ADMM update \eqref{eq_updatelaw_detailed}. It is initialized as $fl_{i,s}(0) = 1$ if the equation \eqref{eq_terminated} is satisfied and $fl_{i,s}(0) = 0$, otherwise.
It is guaranteed that $fl_{i,s}(N) = \min\{fl_{1,s}(0), fl_{2,s}(0), \cdots, fl_{N,s}(0) \}$ for all $i = 1, \dots, N$.
If $fl_{i,s^*}(N) = 1$ at the iteration $s^*$, the agent $\mathcal{S}_i$ knows that all conditions in \eqref{eq_terminated} are satisfied for all $i = 1, \dots, N$. Then, the update law \eqref{eq_updatelaw_detailed} can be terminated. We also set $S^{max}$ as the maximum number of iterations for the update \eqref{eq_updatelaw_detailed}.

\RestyleAlgo{ruled}
\begin{algorithm}
\caption{Distributed implementation of agent $\mathcal{S}_i$ for solving the optimization problem \eqref{eq_distributedOpt}. $\textrm{DistSol}(\textbf{W}_i, \textbf{w}_i,\textbf{U}_i, \textbf{u}_i, \textbf{V}_i, \textbf{v}_i, \{\textbf{U}_{ij}: \mathcal{S}_j \in \mathcal{N}_{\mathcal{S}_i}\}, \mathcal{N}_{\mathcal{S}_i})$}\label{alg_proposed_opt}
\KwData{Matrices $\textbf{W}_i, \textbf{U}_i, \textbf{V}_i, \textbf{U}_{ij} \forall \mathcal{S}_j \in \mathcal{N}_{\mathcal{S}_i}$, vectors $\textbf{w}_i, \textbf{u}_i, \textbf{v}_i$, set of neighbors $\mathcal{N}_{\mathcal{S}_i}$}
 \textit{Initialization:} choose arbitrarily $\textbf{x}_i(0), \textbf{y}_i(0), \boldsymbol{\lambda}_i(0)$ and $\textbf{y}_{ij}(0), \boldsymbol{\lambda}_{ij}(0) \forall \mathcal{S}_j \in \mathcal{N}_{\mathcal{S}_i}$\;
 \For{$s = 0, 1, \dots, S^{max}$}{
  $\textbf{x}_i(s+1) \gets (\ref{eq_updatelaw_detailed}a)$\;
  $\textbf{y}_i(s+1) \gets (\ref{eq_updatelaw_detailed}b)$; $\boldsymbol{\lambda}_i(s+1) \gets (\ref{eq_updatelaw_detailed}c)$\;
  send $\textbf{U}_{ij}\textbf{x}_i(s+1) - \frac{1}{\rho}\boldsymbol{\lambda}_{ij}(s)$ to $\mathcal{S}_j$, $\forall \mathcal{S}_j \in \mathcal{N}_{\mathcal{S}_i}$\;
  receive all $\textbf{U}_{ji}\textbf{x}_j(s+1) - \frac{1}{\rho}\boldsymbol{\lambda}_{ji}(s)$ from $\mathcal{S}_j \in \mathcal{N}_{\mathcal{S}_i}$\;
  $(\textbf{y}_{ij}, \boldsymbol{\lambda}_{ij})(s+1) \gets (\ref{eq_updatelaw_detailed}d-\ref{eq_updatelaw_detailed}e), \forall \mathcal{S}_j \in \mathcal{N}_{\mathcal{S}_i}$\;
  check condition \eqref{eq_terminated} and set $fl_{i,s}(0)$\;
  \For{$\varsigma = 0, 1, \dots, N$}{$fl_{i,s}(\varsigma+1) \gets \eqref{eq_minconsensus}$}  
  \If{$fl_{i,s}(N) == 1$}{stop and output Result\;}
 }
\KwResult{$\textbf{x}_i(s)$}
\end{algorithm}

To conclude this part, we provide Algorithm \ref{alg_proposed_opt} as distributed method for each agent $\mathcal{S}_i$ to find the optimal solution of the optimization problem \eqref{eq_distributedOpt}. We refer the process of running this algorithm as \[\textrm{DistSol}(\textbf{W}_i, \textbf{w}_i,\textbf{U}_i, \textbf{u}_i, \textbf{V}_i, \textbf{v}_i, \{\textbf{U}_{ij}: \mathcal{S}_j \in \mathcal{N}_{\mathcal{S}_i}\}, \mathcal{N}_{\mathcal{S}_i}).\]
This algorithm is fully distributed since every agent is required to use local information belonging to itself or received from its neighbors.
\subsection{Proof of Lemma \ref{lm_TSC}}
Assume that $\hat{\textbf{x}}^{opt}$ is not an optimal solution of the problem \eqref{eq_TSC_compact}. Then there exists at least one vector $\hat{\textbf{x}}^{TSC} = \textrm{col}\{\hat{\textbf{x}}_i^{TSC}: 1 \le i \le N\}$ that satisfies all constraints in \eqref{eq_TSC_compact} and $\sum_{i = 1}^{N}\left(\frac{1}{2}\left(\hat{\textbf{x}}_i^{TSC}\right)^T\textbf{H}_i\hat{\textbf{x}}_i^{TSC} + \textbf{h}_i^T\hat{\textbf{x}}_i^{TSC}\right) < \sum_{i = 1}^{N}\left(\frac{1}{2}\left(\hat{\textbf{x}}_i^{opt}\right)^T\textbf{H}_i\hat{\textbf{x}}_i^{opt} + \textbf{h}_i^T\hat{\textbf{x}}_i^{opt}\right)$.

Let $\hat{u}_{ij}$ be a variable estimated by the agent $\mathcal{S}_i$ for one of its neighbor $\mathcal{S}_j \in \mathcal{N}_{\mathcal{S}_i}$. Consider the following linear algebraic equation
\begin{equation}\label{eq_temp_proof}
\sum_{\mathcal{S}_j \in \mathcal{N}_{\mathcal{S}_i}}\left(\hat{u}_{ij} - \hat{u}_{ji}\right) = \textbf{c}_i^T\left(\textbf{x}_i^{(PC)}-\hat{\textbf{x}}_i^{TSC}\right), \forall i = 1, 2, \dots, N.
\end{equation}
The linear algebraic equation \eqref{eq_temp_proof} consists of $\sum_{i = 1}^{N}|\mathcal{N}_{\mathcal{S}_i}|$ unknown variables, i.e., $\hat{u}_{ij}, \forall \mathcal{S}_j \in \mathcal{N}_{\mathcal{S}_i}, \forall i = 1, 2, \dots, N$, and $N$ linear equations. Here, $|\mathcal{N}_{\mathcal{S}_i}|$ is the number of neighboring agents of the agent $\mathcal{S}_i$. By Assumption \ref{as_graph} and Assumption \ref{as_communication}, we have $\sum_{i = 1}^{N}|\mathcal{N}_{\mathcal{S}_i}| > N$. This guarantees that \eqref{eq_temp_proof} has infinite solutions.

Let $\{u_{ij}^{TSC}: \mathcal{S}_j \in \mathcal{N}_{\mathcal{S}_i}, i = 1, 2, \dots, N\}$ be one solution of the linear algebraic equation \eqref{eq_temp_proof}. Consider the vector $\tilde{\textbf{x}}^{TSC} = \textrm{col}\{\tilde{\textbf{x}}_i^{TSC}: 1 \le i \le N\}$ where $\tilde{\textbf{x}}_i^{TSC}$ is constructed as $\tilde{\textbf{x}}_i^{TSC} = [\hat{\textbf{x}}_i^{TSC}, \textrm{col}\{u_{ij}^{TSC} - u_{ji}^{TSC}: \mathcal{S}_j \in \mathcal{N}_{\mathcal{S}_i}\}]$.
It is no doubt that the vector $\tilde{\textbf{x}}^{TSC}$ satisfies all constraints in (\ref{eq_TSC_tosolve}b-\ref{eq_TSC_tosolve}c). So $\tilde{\textbf{x}}^{TSC}$ is a feasible solution of the problem \eqref{eq_TSC_compact}. Beside that $\tilde{\textbf{x}}^{TSC}$ has the cost value as $\sum_{i = 1}^{N}\left(\frac{1}{2}\left(\hat{\textbf{x}}_i^{TSC}\right)^T\textbf{H}_i\hat{\textbf{x}}_i^{TSC} + \textbf{h}_i^T\hat{\textbf{x}}_i^{TSC}\right)$, which is smaller than the optimal cost value corresponding to $\tilde{\textbf{x}}^{opt}$. This is a contradiction. So $\hat{\textbf{x}}^{opt}$ must be an optimal solution of the problem \eqref{eq_TSC_compact}.
\subsection{Illustrative example}
Consider an urban traffic network shown in Fig. \ref{fig_exampleUTN}. It consists of $4$ internal junctions and $7$ boundary junctions.
\begin{figure}[htb]
\centering
\includegraphics[width=0.3\textwidth]{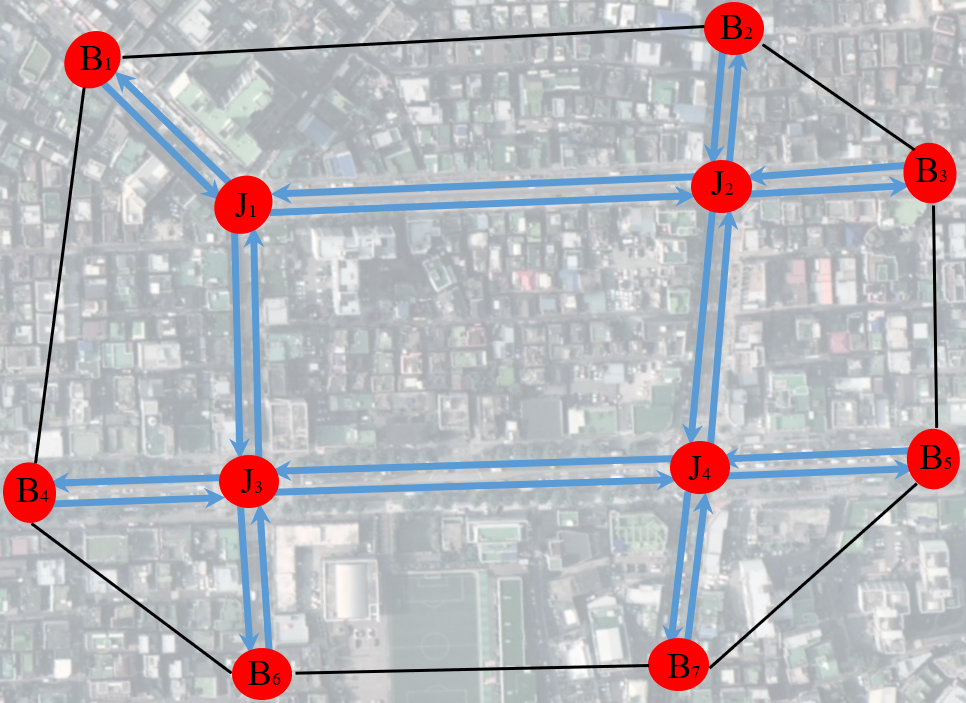}
\caption{An example for urban traffic network. (Source image from Google Earth Map). Red ellipses represent junctions, blue arrows represent roads, and black lines correspond to the perimeter of the UTN.} \label{fig_exampleUTN}
\end{figure}
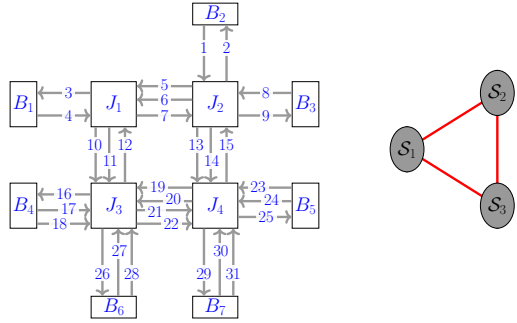
\begin{figure}[htb]
\centering
\scalebox{0.3}{\begin{tikzpicture}[
textnode/.style={rectangle, fill=white, opacity=0.9, minimum size=1mm, text=blue,text opacity=1,font=\fontsize{20}{20}\selectfont},
internalnode/.style={rectangle, draw=black, fill=white,  very thick, minimum size=20mm, text=blue,text opacity=1,font=\fontsize{25}{25}\selectfont},
boundarynode1/.style={rectangle, draw=black, fill=white,  very thick, minimum width=20mm, minimum height=6mm, text=blue,text opacity=1,font=\fontsize{25}{25}\selectfont},
boundarynode2/.style={rectangle, draw=black, fill=white,  very thick, minimum width=6mm, minimum height=20mm, text=blue,text opacity=1,font=\fontsize{25}{25}\selectfont},
framenode/.style={rectangle, fill=white, opacity=0.0, minimum size=150mm, text=blue},
agentnode/.style={ellipse, draw=black, fill=black!40,  very thick, minimum width=6mm, minimum height=20mm, text=black,text opacity=1,font=\fontsize{25}{25}\selectfont},
framenode/.style={rectangle, fill=white, opacity=0.0, minimum size=150mm, text=blue},
]
\node at (2.25,9.5) [framenode] (nFrame) {};
\node at (1,11.5) [internalnode] (n1) {$J_{1}$};
\node at (5.5,11.5) [internalnode] (n2) {$J_{2}$};
\node at (1,7) [internalnode] (n3) {$J_{3}$};
\node at (5.5,7) [internalnode] (n4) {$J_{4}$};
\node at (-3,11.5) [boundarynode2] (nB1) {$B_1$};
\node at (5.5,15.5) [boundarynode1] (nB2) {$B_2$};
\node at (9.5,11.5) [boundarynode2] (nB3) {$B_3$};
\node at (-3,7) [boundarynode2] (nB4) {$B_4$};
\node at (9.5,7) [boundarynode2] (nB5) {$B_5$};
\node at (1,2.5) [boundarynode1] (nB6) {$B_6$};
\node at (5.5,2.5) [boundarynode1] (nB7) {$B_7$};
\node at (14, 9.5) [agentnode] (nS1) {};
\node at (18, 12) [agentnode] (nS2) {};
\node at (18, 7) [agentnode] (nS3) {};
\draw[-,{line width=3pt},red] (nS1.center)--(nS2.center);
\draw[-,{line width=3pt},red] (nS2.center)--(nS3.center);
\draw[-,{line width=3pt},red] (nS3.center)--(nS1.center);
\node at (14, 9.5) [agentnode] (nS1draw) {$\mathcal{S}_1$};
\node at (18, 12) [agentnode] (nS2draw) {$\mathcal{S}_2$};
\node at (18, 7) [agentnode] (nS3draw) {$\mathcal{S}_3$};
\draw[->,{line width=3pt},black!40, transform canvas={xshift=-5mm}] (nB2.south)--(nB2.south|-n2.north);
\draw[<-,{line width=3pt},black!40, transform canvas={xshift=5mm}] (nB2.south)--(nB2.south|-n2.north);

\draw[<-,{line width=3pt},black!40, transform canvas={yshift=5mm}] (nB1.east)->(n1.west);
\draw[->,{line width=3pt},black!40, transform canvas={yshift=-5mm}] (nB1.east)->(n1.west);
\draw[<-,{line width=3pt},black!40, transform canvas={yshift=2mm}] (n1.east)->(n2.west);
\draw[<-,{line width=3pt},black!40, transform canvas={yshift=8mm}] (n1.east)->(n2.west);
\draw[->,{line width=3pt},black!40, transform canvas={yshift=-5mm}] (n1.east)->(n2.west);
\draw[<-,{line width=3pt},black!40, transform canvas={yshift=5mm}] (n2.east)->(nB3.west);
\draw[->,{line width=3pt},black!40, transform canvas={yshift=-5mm}] (n2.east)->(nB3.west);

\draw[->,{line width=3pt},black!40, transform canvas={xshift=-8mm}] (n1.south)--(n1.south|-n3.north);
\draw[->,{line width=3pt},black!40, transform canvas={xshift=-2mm}] (n1.south)--(n1.south|-n3.north);
\draw[<-,{line width=3pt},black!40, transform canvas={xshift=5mm}] (n1.south)--(n1.south|-n3.north);
\draw[->,{line width=3pt},black!40, transform canvas={xshift=-8mm}] (n2.south)--(n2.south|-n4.north);
\draw[->,{line width=3pt},black!40, transform canvas={xshift=-2mm}] (n2.south)--(n2.south|-n4.north);
\draw[<-,{line width=3pt},black!40, transform canvas={xshift=5mm}] (n2.south)--(n2.south|-n4.north);

\draw[<-,{line width=3pt},black!40, transform canvas={yshift=5mm}] (nB4.east)->(n3.west);
\draw[->,{line width=3pt},black!40, transform canvas={yshift=-2mm}] (nB4.east)->(n3.west);
\draw[->,{line width=3pt},black!40, transform canvas={yshift=-8mm}] (nB4.east)->(n3.west);
\draw[<-,{line width=3pt},black!40, transform canvas={yshift=8mm}] (n3.east)->(n4.west);
\draw[<-,{line width=3pt},black!40, transform canvas={yshift=2mm}] (n3.east)->(n4.west);
\draw[->,{line width=3pt},black!40, transform canvas={yshift=-2mm}] (n3.east)->(n4.west);
\draw[->,{line width=3pt},black!40, transform canvas={yshift=-8mm}] (n3.east)->(n4.west);
\draw[<-,{line width=3pt},black!40, transform canvas={yshift=8mm}] (n4.east)->(nB5.west);
\draw[<-,{line width=3pt},black!40, transform canvas={yshift=2mm}] (n4.east)->(nB5.west);
\draw[->,{line width=3pt},black!40, transform canvas={yshift=-5mm}] (n4.east)->(nB5.west);

\draw[->,{line width=3pt},black!40, transform canvas={xshift=-5mm}] (n3.south)--(n3.south|-nB6.north);
\draw[<-,{line width=3pt},black!40, transform canvas={xshift=2mm}] (n3.south)--(n3.south|-nB6.north);
\draw[<-,{line width=3pt},black!40, transform canvas={xshift=8mm}] (n3.south)--(n3.south|-nB6.north);
\draw[->,{line width=3pt},black!40, transform canvas={xshift=-5mm}] (n4.south)--(n4.south|-nB7.north);
\draw[<-,{line width=3pt},black!40, transform canvas={xshift=2mm}] (n4.south)--(n4.south|-nB7.north);
\draw[<-,{line width=3pt},black!40, transform canvas={xshift=8mm}] (n4.south)--(n4.south|-nB7.north);
\node at (5.0,14.0) [textnode] (nL1) {$1$};
\node at (6.0,14.0) [textnode] (nL2) {$2$};

\node at (-1,12) [textnode] (nL3) {$3$};
\node at (-1,11) [textnode] (nL4) {$4$};

\node at (3.25,12.35) [textnode] (nL5) {$5$};
\node at (3.25,11.7) [textnode] (nL6) {$6$};
\node at (3.25,11) [textnode] (nL7) {$7$};

\node at (7.75,12) [textnode] (nL8) {$8$};
\node at (7.75,11) [textnode] (nL9) {$9$};

\node at (0.15,9.7) [textnode] (nL10) {$10$};
\node at (0.85,9.0) [textnode] (nL11) {$11$};
\node at (1.5,9.7) [textnode] (nL12) {$12$};

\node at (4.65,9.7) [textnode] (nL13) {$13$};
\node at (5.35,9.0) [textnode] (nL14) {$14$};
\node at (6.0,9.7) [textnode] (nL15) {$15$};

\node at (-1.2,7.55) [textnode] (nL6) {$16$};
\node at (-1,6.85) [textnode] (nL17) {$17$};
\node at (-1.4,6.2) [textnode] (nL18) {$18$};

\node at (2.95,7.85) [textnode] (nL19) {$19$};
\node at (3.65,7.3) [textnode] (nL20) {$20$};
\node at (2.85,6.75) [textnode] (nL21) {$21$};
\node at (3.55,6.2) [textnode] (nL22) {$22$};

\node at (7.4,7.85) [textnode] (nL23) {$23$};
\node at (8.0,7.25) [textnode] (nL24) {$24$};
\node at (7.75,6.5) [textnode] (nL25) {$25$};

\node at (0.5,4.0) [textnode] (nL26) {$26$};
\node at (1.25,5.0) [textnode] (nL27) {$27$};
\node at (1.8,4.0) [textnode] (nL28) {$28$};

\node at (5.0,4.0) [textnode] (nL29) {$29$};
\node at (5.75,5.0) [textnode] (nL30) {$30$};
\node at (6.3,4.0) [textnode] (nL31) {$31$};
\end{tikzpicture}}
\caption{The left figure is the graph representation for the UTN in Fig. \ref{fig_exampleUTN_graph}. The right figure is the communication graph for multiple agents controlling the UTN.} \label{fig_exampleUTN_graph}
\end{figure}
\begin{table}
\begin{center}
\captionof{table}{Sequences of traffic signal phases corresponding to internal junctions in Fig. \ref{fig_exampleUTN}.}
\scalebox{0.8}
{\begin{tabular}{c|c|c|c|c|c}
\hline
Type & Junction & $1^{st}$ phase & $2^{nd}$ phase & $3^{rd}$ phase & $4^{th}$ phase\\
\hline
\rule{0pt}{15pt} 1 & $J_1$ & \includegraphics[width=0.025\textwidth, angle=90]{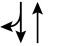} & \includegraphics[width=0.025\textwidth]{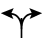}  & \includegraphics[width=0.025\textwidth, angle=90]{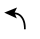} & none \\
\hline
\rule{0pt}{15pt} 2 & $J_2$ & \includegraphics[width=0.025\textwidth]{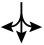} & \includegraphics[width=0.025\textwidth, angle=90]{4_phasetype3.png} & \includegraphics[width=0.025\textwidth]{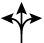} & \includegraphics[width=0.025\textwidth, angle=90]{4_phasetype4.png}\\
\hline
\rule{0pt}{15pt} 3 & $J_3, J_4$ & \includegraphics[width=0.025\textwidth]{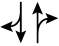} & \includegraphics[width=0.025\textwidth]{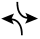}  & \includegraphics[width=0.025\textwidth, angle=90]{4_phasetype1.png} & \includegraphics[width=0.025\textwidth, angle=90]{4_phasetype2.png} \\
\hline
\end{tabular}}
\label{tbl_sequencePhases}
\end{center}
\end{table}
\begin{table}
\begin{center}
\captionof{table}{Active road links in traffic signal phases.}
\scalebox{0.8}
{\begin{tabular}{c|c|c|c|c}
\hline
Junction & $1^{st}$ phase & $2^{nd}$ phase & $3^{rd}$ phase & $4^{th}$ phase\\
\hline
\rule{0pt}{15pt} $J_1$ & $4, 5$ & $12$ & $6$ & none \\
\hline
\rule{0pt}{15pt} $J_2$ & $1$ & $7$ & $15$ & $8$\\
\hline
\rule{0pt}{15pt} $J_3$ & $10, 28$ & $11, 27$  & $18, 19$ & $17, 20$ \\
\hline
\rule{0pt}{15pt} $J_4$ & $13, 31$ & $14, 30$ & $22, 23$ & $21, 24$ \\
\hline
\end{tabular}}
\label{tbl_LinkPhases}
\end{center}
\end{table}
Assume that the sequences of traffic signal phases in internal junctions are given in TABLE. \ref{tbl_sequencePhases}.
Then we have the graph representation of this UTN as the left figure in Fig. \ref{fig_exampleUTN_graph}. In which, the gray arrows illustrate road links in $\mathcal{R}$, the squares represent internal junctions in $\mathcal{J}^I$, and the rectangles correspond to the boundary junctions in the set $\mathcal{J}^B$. The graph $\mathcal{T} = (\mathcal{J}, \mathcal{R})$ includes the road link set $\mathcal{R} = \{1, 2, \dots, 31\}$ and the junction set $\mathcal{J} = \mathcal{J}^B \cup \mathcal{J}^I$. In which, $\mathcal{J}^B = \{B_{1}, B_{2}, B_{3}, B_{4}, B_{5}, B_{6}, B_{7}\}$ is the set of boundary junctions and $\mathcal{J}^I = \{J_{1}, J_{2}, J_{3}, J_{4}\}$ is the set of internal junctions. Road links $6, 11, 14, 17, 20, 21, 24, 27$, $30$ are for turning left only directions of their corresponding roads.

For the road links $10, 11$ and $12$, we have $\sigma(10) = \sigma(11) = J_1, \tau(10) = \tau(11) = J_3$, $\mathcal{N}_{10}^+ = \mathcal{N}_{11}^+ = \{4, 5, 6\}$, $\mathcal{N}_{10}^- = \{16, 26\}$, $\mathcal{N}_{11}^- = \{21, 22\}$, $\sigma(12) = J_3, \tau(12) = J_1$, $\mathcal{N}_{12}^+ = \{17, 19, 28\}$, $\mathcal{N}_{12}^- = \{3, 7\}$.
For the intersection $J_1$, we have $\mathcal{R}_{J_1}^{in} = \{4, 5, 6, 12\}$ and $\mathcal{R}_{J_1}^{out} = \{3, 7, 10, 11\}$.

TABLE. \ref{tbl_LinkPhases} indicates the road links corresponding to the traffic signal phases of the internal junctions.
Denote by $p_i^{J_v}$ the $i^{th}$-traffic signal phase of the internal junction $J_v \in \mathcal{J}^I$.
Let $g_{p_i}^{J_v}(t)$ and $\tilde{g}_z(t)$ be the split of the traffic signal phase $p_i^{J_v}$ and the green time assigned to the road link $z$ in the $t^{th}$-cycle, respectively. Consider the junction $J_3$. We have the following inequalities:
\[g_{p_1}^{J_3}(t) + g_{p_2}^{J_3}(t) + g_{p_3}^{J_3}(t) + g_{p_4}^{J_3}(t) \le T-L_{J_3},\]
\[0 \le \tilde{g}_{10}(t) \le g_{p_1}^{J_3}(t), 0 \le \tilde{g}_{28}(t) \le g_{p_1}^{J_3}(t),\]
\[0 \le \tilde{g}_{11}(t) \le g_{p_2}^{J_3}(t), 0 \le \tilde{g}_{27}(t) \le g_{p_2}^{J_3}(t),\]
\[0 \le \tilde{g}_{18}(t) \le g_{p_3}^{J_3}(t), 0 \le \tilde{g}_{19}(t) \le g_{p_3}^{J_3}(t),\]
\[0 \le \tilde{g}_{17}(t) \le g_{p_4}^{J_3}(t), 0 \le \tilde{g}_{20}(t) \le g_{p_4}^{J_3}(t).\]
Though both road links $11$ and $27$ belong to the second traffic signal phase of the junction $J_3$, i.e., $\mathcal{P}_{11} = \mathcal{P}_{27} = \{p_2^{J_3}\}$, their assigned green times are not required to be the same or equal to $g_{p_2}^{J_3}(t)$. That means it is possible to have $\tilde{g}_{11}(t) \neq \tilde{g}_{27}(t)$ and $\tilde{g}_{11}(t) < g_{p_2}^{J_3}(t)$, $\tilde{g}_{27}(t) < g_{p_2}^{J_3}(t)$.

Assuming the UTN is divided into three subnetworks $\mathcal{S}_1$, $\mathcal{S}_2$ and $\mathcal{S}_3$ where $\mathcal{J}_1^I = \{J_1, J_3\}, \mathcal{J}_1^B = \{B_1, B_4, B_6\}$, $\mathcal{J}_2^I = \{I_2\}, \mathcal{J}_2^B = \{B_2, B_3\}$, and $\mathcal{J}_3^I = \{J_4\}, \mathcal{J}_3^B = \{B_5, B_7\}$.
Then we have the sets of internal road links as
\[\mathcal{R}_{1} = \{3, 4, 10, 11, 12, 16, 17, 18, 26, 27, 28\},\]
\[\mathcal{R}_{2} = \{1, 2, 8, 9\}, \textrm{ and } \mathcal{R}_{3} = \{23, 24, 25, 29, 30, 31\}.\]
The sets of connecting road links among subnetworks are $\mathcal{R}_{12} = \{7\}, \mathcal{R}_{21} = \{5, 6\}, \mathcal{R}_{13} = \{21, 22\}, \mathcal{R}_{31} = \{19, 20\}, \mathcal{R}_{23} = \{13, 14\}, \mathcal{R}_{32} = \{15\}$.
Let each agent $\mathcal{S}_i$ controlling the subnetwork $\mathcal{S}_i, \forall i = 1, 2, 3$. We have the communication graph as shown in the right figure in Fig. \ref{fig_exampleUTN}. In which, the red lines present the communication channel among agents. The neighbor sets of agents are given as $\mathcal{N}_{\mathcal{S}_1} = \{\mathcal{S}_2, \mathcal{S}_3\}$, $\mathcal{N}_{\mathcal{S}_2} = \{\mathcal{S}_1, \mathcal{S}_3\}$, and $\mathcal{N}_{\mathcal{S}_3} = \{\mathcal{S}_1, \mathcal{S}_2\}$.
\bibliographystyle{IEEEtran}
\bibliography{mylib}

\end{document}